\documentclass[11pt]{article}
\usepackage[letterpaper, left=0.85in, top=0.85in, right=0.85in, bottom=0.85in,nohead,includefoot, verbose,ignoremp]{geometry}
\usepackage{color,charter,graphicx,amsmath,amssymb,amsfonts,amsthm,hypernat,setspace,multirow,caption}
\usepackage[usenames,dvipsnames,x11names]{xcolor}
\usepackage[center]{titlesec}

\titleformat{\subsection}
  {\normalfont\large\bfseries}{\thesubsection}{1em}{}
\titlelabel{\thetitle.\quad}
\usepackage[pdftex,pagebackref=true]{hyperref}
\hypersetup{
colorlinks,
linkcolor=RoyalBlue1,
urlcolor=RoyalBlue2,
citecolor=RoyalBlue3
}

\def\bSigma{\mbox{\boldmath$\Sigma$}}
\def\bmu{\mbox{\boldmath$\mu$}}
\def\bPhi{\mbox{\boldmath$\Phi$}}
\def\bPsi{\mbox{\boldmath$\Psi$}}
\def\bomega{\mbox{\boldmath$\omega$}}
\def\bTheta{\mbox{\boldmath$\Theta$}}

\def\bgamma{\mbox{\boldmath$\gamma$}}
\def\bGamma{\mbox{\boldmath$\Gamma$}}
\def\bepsilon{\mbox{\boldmath$\epsilon$}}
\def\x{\mbox{\boldmath$x$}}
\def\e{\mbox{\boldmath$e$}}
\def\X{\mbox{\boldmath$X$}}
\def\y{\mbox{\boldmath$y$}}
\def\Y{\mbox{\boldmath$Y$}}
\def\R{\mbox{\boldmath$R$}}
\def\O{\mbox{\boldmath$O$}}
\def\V{\mbox{\boldmath$V$}}
\def\U{\mbox{\boldmath$U$}}
\def\S{\mbox{\boldmath$S$}}
\def\D{\mbox{\boldmath$D$}}
\def\Q{\mbox{\boldmath$Q$}}
\def\P{\mbox{\boldmath$P$}}
\def\I{\mbox{\boldmath$I$}}
\def\A{\mbox{\boldmath$A$}}
\def\B{\mbox{\boldmath$B$}}
\def\K{\mbox{\boldmath$K$}}
\def\H{\mbox{\boldmath$H$}}
\def\G{\mbox{\boldmath$G$}}
\def\m{\mbox{\boldmath$m$}}
\def\M{\mbox{\boldmath$M$}}
\def\V{\mbox{\boldmath$V$}}
\def\w{\mbox{\boldmath$w$}}
\def\0{\mbox{\boldmath$0$}}
\def\1{\mbox{\boldmath$1$}}
\def\m{\mbox{\boldmath$m$}}
\def\w{\mbox{\boldmath$w$}}
\def\cM{\mathcal{M}}
\def\cZ{\mathcal{Z}}
\def\cS{\mathcal{S}}
\def\cE{\mathcal{E}}
\def\cV{\mathcal{V}}
\def\cG{\mathcal{G}}
\def\cv{\mathcal{U}}
\def\cN{\mathcal{N}}

\def\eqno#1{equation~(\ref{eq:#1})}

\newtheorem{thm}{Theorem}

\begin{document}

\title{Models of Random Sparse Eigenmatrices \\ \& \\  Bayesian Analysis of Multivariate Structure}
 
\author{Andrew Cron\footnote{$84.51^\circ$ and Duke University.
				 \href{mailto:andrew.j.cron@gmail.com}{andrew.j.cron@gmail.com}}\,\,  \& 
Mike West\footnote{Duke University. \href{mailto:mw@stat.duke.edu}{mw@stat.duke.edu}}
	       \footnote{Research partly supported by grants from the
National Science Foundation [DMS-1106516] and the National 
\newline\indent \quad Institutes of Health [1RC1-AI086032].  Any opinions, findings and conclusions or recommendations expressed \newline\indent \quad in
this work are those of the authors and do not necessarily reflect the views of the NSF or NIH.}
}

\maketitle

\begin{abstract}
We discuss probabilistic models of random covariance structures defined
by  distributions over sparse eigenmatrices.
The decomposition of orthogonal matrices in terms of Givens rotations
defines a natural, interpretable framework for defining distributions on
sparsity structure of random eigenmatrices.   We explore theoretical
aspects and implications for conditional independence structures arising
in multivariate Gaussian models, and
discuss connections with sparse PCA, factor analysis and Gaussian graphical models.
Methodology includes model-based exploratory
data analysis and Bayesian analysis via reversible
jump Markov chain Monte Carlo.  A simulation study examines the ability to
identify sparse multivariate structures compared to the  benchmark graphical modelling approach.
Extensions to multivariate normal mixture models with additional measurement errors move into the
framework of latent structure analysis of broad practical interest.  We explore the
implications and utility of the new models with summaries of a detailed
applied study of a $20-$dimensional breast cancer genomics data set.


\end{abstract}

\newpage

\section{Introduction}
We are interested in Bayesian modelling approaches to sparsity in
variance and precision matrices in multivariate normal
distributions. With interests in parsimony and scalability of analyses
of multivariate data in models such as Gaussian mixtures for
classification, priors that encourage sparse component covariance
patterns are increasingly key as dimension increases.  New modelling
frameworks also need to enable efficient computational methods for
model fitting, which can otherwise be a barrier to application.

Among recent related developments, traditional sparsity priors from
model selection in regression have been exploited in sparse extensions
of Bayesian factor
analysis~\cite{West2003,Carvalho2008,Yoshida2010}, and in
complementary approaches using Gaussian graphical
models~\cite{Jones2005,Dobra2011,Rodriguez2011}.  The developments in the
current work  represent natural extensions of the thinking behind these models--
building sparsity into variance or precision matrices-- while
naturally linking and bridging between factor models and graphical
models.

The new \lq\lq sparse Givens'' models introduced and developed here
arise from new theory of random sparse eigenmatrices; these define
eigenstructure of  variance and precision matrices, and so induce
new classes of priors over Gaussian graphical models. Compared to factor analysis, we avoid
the  assumption of a reduced dimensional latent factor
structure, and the choices it involves.  Our new models arise from an
inherent theoretical feature of eigenmatrices, rather than
hypothesized model structures.  We also face fewer challenges in
hyper-parameter specification and tuning to fit models. Our models can
in fact be viewed as full-rank factor models with sparse, square
factor loadings matrices.  Additional related work has explored new
classes of priors over variance matrices through varying
parametrizations, such as partial correlations or Cholesky
decompositions~\cite{Daniels2009,Lopes2010c}, that could be
extended with sparsity priors. Some such extensions to time series 
contexts~\cite{Nakajima2013,GruberWest2015BA} 
show the utility of various Cholesky-style approaches. 
Our approach relates to this general literature in that
it uses an inherent theoretical property of eigenmatrices that
naturally defines the reparametrization as well as an underlying set
of parameters that, when set to zero, define parsimonious models.

Section~\ref{sec:eigtheory} introduces the theoretical and modelling
ideas; the approach is based on the Givens rotation representation of
full-rank eigenmatrices~\cite{Anderson1987}. We describe how this can
be exploited to define new classes of random sparse eigenmatrices, and
relate these to decomposable graphical models.
Section~\ref{sec:priors} considers prior specification over variance
matrices using this new parametrization, in the context of normal
random samples.  Section~\ref{sec:explore} discusses properties of the
likelihood and aspects of exploratory data analysis that give insights
into sparsity structure of eigenmatrices in our framework, with an
example using a $20-$dimensional gene expression data set.
Section~\ref{sec:mcmc} discusses full Bayesian model fitting using a
customized reversible jump Markov Chain Monte Carlo approach.
We make a detailed, simulation-based comparison
with traditional Gaussian graphical modelling (GGM) in~\ref{sec:simstudy}.
Section~\ref{sec:mixtures} discusses embedding the basic model into
more practicable contexts involving measurement errors and normal
mixture models. That section concludes with a detailed example using
breast cancer gene expression data, where underlying components relate
to known, broad and intersecting cancer subtypes with expected
sparsity in dependence, and conditional dependence patterns of subsets
of the genes.  Section~\ref{sec:end} concludes with additional
comments and potential extensions.

\section{Structure and Sparsity in Eigenmatrices
\label{sec:eigtheory}}
We discuss Givens representations of variance matrices, introduce the general
idea of sparsity modelling in this context, and explore aspects of the theoretical
structures that emerge under priors over the resulting models.

\subsection{Givens Rotator Product Representation}
Consider a random $q-$vector $\x$ with variance matrix $\V=Var(\x).$ The
spectral representation\index{spectral representation} 
(principal component decomposition) is
$\V=\R\D\R'$ where $\R$ is the $q\times q$ orthogonal matrix of
eigenvectors-- the eigenmatrix-- and
$\D=\textrm{diag}(d_1,\ldots,d_q)$ is the matrix of non-negative
eigenvalues. The corresponding precision matrix is $\K=\R\A\R'$ with
$\A=\D^{-1} = \textrm{diag}(a_1,\ldots,a_q).$ The general 
Givens rotator product\index{Givens rotator product} 
representation of $\R$~\cite{Anderson1987,Yang1994}
is
\setlength\arraycolsep{1.2pt}
\begin{equation} \label{eq:decomp}
\begin{array}{lrl}
\R = \O_{1,2}(\omega_{1,2})&\O_{1,3}(\omega_{1,3})\ldots \O_{1,q}(\omega_{1,q})& \times \, \\
&\O_{2,3}(\omega_{2,3})\ldots \O_{2,q}(\omega_{2,q})& \times \, \\
&\vdots\ \ \ \ \ \ \ & \\
&\O_{q-1,q}(\omega_{q-1,q})&\times\ \Q
\end{array}
\end{equation}
where $\Q$ is diagonal with elements $\pm 1,$ and each
$\O_{i,j}(\omega_{i,j})$ is a 
Givens rotation matrix\index{Givens rotation matrix}
\begin{equation}\label{eq:givens} \O_{i,j}(\omega_{i,j})\ = \
\bordermatrix{ 		&	    &  i  &    &  j   &   \cr
                             &       I & 0                 & 0 & 0 & 0 \cr
		      i     	& 0 & \phantom{-}\cos(\omega_{i,j}) & 0 & \sin(\omega_{i,j}) & 0 \cr
				&	 0 & 0 & I & 0 & 0 \cr
		      j		&	 0&  -\sin(\omega_{i,j}) & 0 & \cos(\omega_{i,j}) & 0  \cr
				&	 0 & 0 & 0 & 0 & I
}
\end{equation}
for some rotator angles $\omega_{i,j}$, ($i=1:q, \, j=i+1:q).$
Some comments and notation follow.
\begin{itemize}

\item  The angles $\omega_{i,j} $ lie in $ (-\pi/2, \pi/2].$   Write $\bomega$ for the set of these
$m=q(q-1)/2$ angles.

\item This decomposition of $\R$ into $m$ angles is unique and linked to the specific order of the
variables in $\x.$

\item For our goal of covariance modelling, note that $\Q$ cancels in
$\R\D\R'$; hence, $\Q$ plays no role and we set $\Q=\I_q$ with no loss when focused on modelling
variance matrices via this decomposition.

\item Covariance patterns in $\V$ can be viewed as successively
  built-up by pairwise rotations of initial uncorrelated random
  variables.  Take a $q-$vector $\e$ with $Var(\e)=\D; $ then
  dependencies are defined by successive left multiplication of $\e$
  by the rotator matrices: first by $ \O_{q-1,q}(\omega_{i,j}),$ then
  $ \O_{q-2,q}(\omega_{i,j}), $ and so on up to $
  \O_{1,2}(\omega_{i,j})$ to define $\x=\R\e$ (assuming $\Q=\I_q$ as
  noted).

\item If $\omega_{i,j}=0$ for any $(i,j),$ then $\O_{i,j}(0)=\I_q$ and
  that rotation has no contribution to the build-up of dependencies
  and is effectively removed from the representation of \eqno{decomp}.

\item If $\omega_{i,j}=\pi/2$ for any $(i,j),$ then
$\O_{i,j}(\pi/2)\M\O_{i,j}'(\pi/2)$  permutes rows $i$ and $j$, and columns $i$ and $j$,
of any square matrix $\M$ and hence does not affect the sparsity of $\M.$

\item The spectral representations of $\V$ and $\K$ are unique only up
  to permutations of the columns  $\R,$ i.e., reordering of the
  eigenvalues.  Any reordering of the eigenvalues will generate a
  decomposition as in \eqno{decomp} but with different values of the
  rotator angles. For identification, therefore, we will constrain to
  $d_1>d_2>\cdots d_q.$ For variance matrices in models of data
  distributions, the $d_i$ will be distinct so a strict ordering can
  be assumed.

\end{itemize}

 \subsection{Sparse Givens Models}

The general representation above reparametrizes $\V$ to the
$m=q(q-1)/2$ angles in $\bomega$ and the $q$ eigenvalues in $\D$. We note above
the role of  zero angles, and this opens the path to defining 
sparse Givens models\index{sparse Givens model}, i.e., 
products of fewer than the full set of rotators defining  a resulting 
sparse eigenmatrix\index{sparse eigenmatrix}:  
if a large number of the angles are zero, then $\R$ will become sparse.  This can
induce a sparse  variance matrix\index{sparse  variance matrix} $\V$ and, equivalently, a 
sparse precision matrix\index{sparse precision matrix} $\K$ as a result.

Let $\cM = \{ (i,j); \ i=1:(q-1), j=(i+1):q \}$ with $|\cM|=q(q-1)/2:=m$. Then
\eqno{decomp} is compactly written as
$\R = \prod_{k=1}^m \O_{m_k}(\omega_{m_k})$
where $m_k$ is the $k^\textrm{th}$ pair of dimensions in $\cM$.
Now allow exact zeros in $\bomega.$
Define a sparsity defining index sequence $\cS=\{(i,j)\in
\cM : \omega_{i,j} = 0 \}$ with cardinality $|\cS|=s$, and set
$\cZ=\cM \setminus \cS$ with size $z=m-s$.  In words, $\cZ$ is a sequence of
$z\leq m$ ordered pairs $(i_k,j_k)$  denoting the relevant, non-identity Givens rotation matrices in
\eqno{decomp} and
\begin{equation} \label{eq:decompsparse}
\R=\prod_{k=1}^z \O_{Z_k}(\omega_{i_k,j_k}).
\end{equation}
Assuming that priors support exact zeros in $\bomega,$ a primary modelling goal is
then to learn $\cZ$ and
the corresponding non-zero angles.

Among the features of this approach 
is that we are able to model full-rank,
orthogonal matrices with a parsimonious set of angles, and we maintain
the computational convenience of the full-rank spectral
parametrization   when inverting $\V$ and $\K.$ This is especially useful
in evaluating density functions in Metropolis Hastings acceptance ratios and, later, in
computing normal mixture classification probabilities.

\subsection{Conditional Independence Graphs}

The process of successively building dependencies by adding rotators
(from right to left) in \eqno{decompsparse} induces ties between the
variables whose variance matrix is the resulting $\V.$ The resulting
structure of $\K=\V^{-1}$ connects to Gaussian graphical\index{Gaussian graphical model} modelling~\cite{lauritzen96,Dobra2004,Jones2005,carvalho:west:07,Dobra2011,Rodriguez2011}.
\index{conditional independence graph}

View  the $q$ variables in $\x$ as
nodes of a  graph in which conditional independencies are represented by lack of edges between
node pairs. Specifically, this is the undirected graph $\cG=(\cV,\cE)$ with the $q$
nodes, or vertices, in the vertex set $\cV=\{ 1:q \}$; two vertices $(i,j)$ are connected
by an edge in the graph if, and only if,
$K_{i,j}\neq0$ where $K_{i,j}$ is the $(i,j)-$element of $\K.$ The edge set is
$\cE = \{ (i,j): \, K_{i,j}\ne0 \}. $

Any precision matrix $\K_0$ having some off-diagonal zero elements has an
implied graph $\cG_0.$ Now take
$\K_1 = \O_{i,j}\K_0\O_{i,j}'$ where $\omega_{i,j}\neq0$ and
$\omega_{i,j}\neq\pi/2,$ with implied graph $\cG_1$.
Notice that left multiplication of $\R_0$ by $\O_{i,j}$ simply replaces the
  $i^{th}$ and $j^{th}$ rows of  with a linear combination of the
  two. Therefore, the indices of the non-zero elements of the $i^{th}$
  and $j^{th}$ rows of $\R_1$ are the union of the indices
  of the $i^{th}$ and $j^{th}$ rows of $\R_0$.     Similar comments apply to right multiplication. As
a result,  the sparsity pattern of $\K_1$ is the same as that of $\K_0$ except in rows and
columns $i$ and $j$. Specifically, those rows and columns have
sparsity indices that are the unions of the those in $\K_0.$
This shows that  the additional rotator $\O_{i,j}$ maps the
graph $\cG_0 = (\cV,\cE_0)$  to  $\cG_1 = (\cV,\cE_1)$ as follows.
With  $\cN_j(i)=\{(j,k):\, (i,k)\in \cE_0\}$,  then $\cE_1=\cE_0\cup(i,j)\cup\cN_i(j)\cup
\cN_j(i)$. In words, $\cG_1$ takes $\cG_0$, connects $i$ and
$j$, and unions their neighborhoods.

This structure also generates  constructive insights into the nature of the graphical models
so defined.   It shows that adding a new rotator to an existing sparse Givens model merges the
complete subgraphs  (cliques) in which the rotators pair $(i,j)$ reside into one larger clique.
Starting at an empty graph, this leads to graphs that are decomposable, formally shown as follows.

\begin{thm} \label{th:decomp} The conditional independence graph $\cG$
  implied by a sparse $\K=\R\A\R'$ under  $\cZ$ is a decomposable  graph.\index{decomposable graphical model}
\end{thm}

\begin{proof}
  It is enough to show that $\cG$ has a perfect elimination ordering; that is, an ordering of the vertices
  of the graph such that, for each vertex $v\in \cV$, the
  neighbors of $v$ that occur after $v$ in the order form a clique~\cite{Fulkerson1965}. We do this by
   induction, beginning with no rotations: $\R_0=\I_q$ and $\K_0=\A.$
  This implies that $\cG_0$ is the empty graph and the perfect
  elimination ordering is trivial. For the inductive step, assume
  that an ordering exists for the graph implied by a current set of 
 Givens rotations defining $\R_0$ and $\cG_0.$ Now take
$\K_1 = \O_{i,j}\K_0\O_{i,j}'$ where $\omega_{i,j}\neq0$ and
$\omega_{i,j}\neq\pi/2,$ with implied graph $\cG_1$.
  Note that there is no loss of generality here;
  $\omega_{i,j}=\pi/2$ would imply simply swapping the $i$ and $j$ rows
  and columns of $\K_0$ to make $\K_1,$ and so always yields another decomposable graph.  It is
  now enough to show that $\cG_1$ has a perfect elimination
  ordering.

  Start with the ordering of $\cG$ given by
  $\cv=\{ v_1,v_2, \ldots,v_q\} \equiv \{ 1,2,\ldots,q\}. $  Now take the ordering for
  $\cG_1$ to be
  $\tilde{\cv}=\{v_1,\dots,v_{i-1},\dots,v_{i+1},\dots,v_{j-1},v_{j+1},\dots,v_q,v_i,v_j\}$. It
  is enough to show that $\tilde{\cv}$ is a perfect elimination
  ordering for $\cG_1$. Take $v^\star\in\tilde{\cv}$ and
  let $\eta$ be $v^\star$ and its neighbors that occur after $v^\star$
  in $\tilde{\cv}$ . We need to show that $\eta$ forms a clique in
  $\cG_1$.  If $v^\star=v_1,v_i,\textrm{or }v_j$, this is
  trivial. If $v^\star$ is not a neighbor of either $v_i$ or $v_j$ in
  $\cG_0$, the rotation has no effect on the neighborhood of $v^\star$
  and $\eta$ is a clique. Now suppose that $v^\star$ is a neighbor
  of $v_i$ in $\cG_0$. Due to our construction of $\cG_1$,
  the neighbors of $v^\star$ in $\cG_0$  become neighbors of $v_i$ in $\cG_1.$  Since
  $\eta\backslash v_i$ is a clique in $\cG_0$,  then $\eta$ will remain a clique
  in $\cG_1$. Since $v_i$ and $v_j$ were moved to the end of
  the ordering and $\eta\backslash v_i$ comes after $v^\star$ in
  $\tilde{\cv}$ by the inductive hypothesis, then $\tilde{\cv}$ is a
  perfect elimination ordering of $\cG_1$.
\end{proof}

Note that the above concerns general, unrestricted values of the
non-zero angles. Furthermore, this applies to any ordering of the
rotators where \eqno{decomp} is a special case.  There are sparse
precision matrices whose graphs are decomposable but that do not have
a sparse Givens representations for their eigenmatrices.  These arise,
in particular, in parametric models where the variance and precision
matrix are initially defined as functions of lower dimensional
parameters to begin; in such cases, the resulting eigenmatrices are
inherently structured and typically not sparse, even though the
precision matrices are sparse.  The simplest example is that of the
dependence structure for a set of $q$ consecutive values of a
stationary, linear, Gaussian first-order autoregressive process.
There $\K$ is tri-diagonal, and neither $\R$ nor $\V$ is sparse.
While $\R$ has the Givens representation, all $m$ angles are required
and they are deterministically related.

In the next section we define priors for the rotator angles $\bomega.$   This includes
conditional priors for the effective angles-- excluding values of $0$ and $\pi/2$--
under which these angle are a random sample from a continuous distribution.  In such cases,
which can be regarded as all practicable cases for applied data analysis, we find a surprising
connection between sparse graphical models and sparse factor models; that is, they coincide in this
new sparse Givens approach.

\begin{thm} \label{th:covsparse}   If the angles $\omega_{i,j} \in \cZ$ defining
a sparse eigenmatrix are a random sample from a continuous distribution, then
the resulting patterns of zeros in
$\V$ and $\K=\V^{-1}$ are the same with probability one.
\end{thm}

\begin{proof}
For any $(i,j)$ pair,
\[V_{i,j} = \sum_{k=1}^q d_k R_{i,k} R_{j,k}\textrm{ and } K_{i,j}
= \sum_{k=1}^q a_k R_{i,k} R_{j,k}.\]
Therefore, zero values of $V_{i,j}$ and $K_{i,j}$ follow when
\[\forall k\in\{1,\dots,q\}\; R_{i,k}=0\textrm{ or }R_{j,k}=0.\]
However, any other case giving $V_{i,j}=0$
requires specific values of $\D$, and/or specific
relationships among elements of $\V$ and $\D$ defining the deterministic constraint that the
above sum be zero.  Such a constraint will not yield $K_{i,j}=0$ under a continuous prior over the angles.
\end{proof}

\section{Sparsity Priors on Eigenstructures
\label{sec:priors}}

We overlay the theoretical framework above with priors that define
interesting theoretical models of random variance matrices as well as
the specifications necessary for Bayesian analysis.

\subsection{Class of Priors}\index{sparsity prior}

 We specify priors that give positive probability to zero values among the angles, allowing
 row and column flips  via angles of $\pi/2$, and that otherwise draw angles independently from
 a continuous distribution.  Specifically, the $m$ angles $\omega_{i,j}$ are a random sample from
 a distribution with density
 \begin{equation} \label{eq:prior}
p(\omega) = \beta_{\frac{\pi}{2}}I(\omega=\pi/2) +
(1-\beta_{\frac{\pi}{2}})\beta_0I(\omega=0) \\
+ (1-\beta_{\frac{\pi}{2}})(1-\beta_0)p_c(\omega)
\end{equation}
where $I(\cdot)$ is the indicator function and $p_c(\omega)$ a continuous density on $(-\pi/2,\pi/2).$

Since $\omega_{i,j}=\pi/2$ does not effect the sparsity of $R$ and is
needed for permuting the effects of the eigenvalues as discussed
earlier, we do not want to penalize permutations in the same way as
other non-zero angles. We specify the prior in three stages. First
with probability $\beta_{\frac{\pi}2},$ $\omega_{i,j}=\pi/2$ to
complement the constraint on eigenvalues $d_i$ being ordered.  Then,
for angles that do not induce a permutation, we  allow zero
values with a non-zero conditional probability $\beta_0.$ Finally,
conditional on $\omega_{i,j}\ne 0$ or $\pi/2,$ it follows a
continuous prior $p_c(\cdot).$

There are various choices of the continuous prior component $p_c(\cdot)$.
Our examples here use a specific form that seems relevant for use as a routine, namely
\begin{equation} \label{eq:ctsprior}
p_c(\omega) =c(\kappa)\exp\left\{\kappa\cos^2(\omega)\right\}I(|\omega|<\pi/2)
\end{equation}
where $\kappa>0$ and $c(\kappa) $ is a normalizing constant.  In Bayesian analyses
via reversible jump MCMC methods we need the value of $c(\kappa)$ and
it can be easily evaluated using any standard numerical integration technique.  This
prior  is unimodal and symmetric about zero, so
represents appropriate centering relative to the \lq\lq null
hypothesis" value at zero.  The prior concentrates more around zero
for larger values of $\kappa,$ while $\kappa\to 0$ leads to the
limiting uniform distribution on $(-\pi/2,\pi/2).$ The specific
mathematical form is also suggested by the forms of conditional
likelihood functions for angles in normal models, as noted below in
Section~\ref{sec:explore}.

The prior is completed by specifying a distribution for the
eigenvalues $\D=\textrm{diag}(d_1,d_2,\dots,d_q)$ of $\V.$ As
discussed above, we take them ordered as $d_1>d_2>\cdots >d_q.$ The
natural, conditionally conjugate class of priors takes the $d_j$ as
ordered values of $q$ independent draws from an inverse gamma
distribution: given some chosen hyperparameters $(\eta_1,\eta_2),$
draw $d_i^{-1} \sim Ga(\eta_1/2,\eta_2/2)$ independently then impose
the ordering.

A specified prior over $(\bomega,\D)$ leads to the implied prior over
$\V$ and $\K,$ and within that a prior over the sparsity structure
that relates to the random graphical model induced.  Simulation of
$p(\bomega,\D)$ yields simulations from the latter. One aspect of
interest is to understand how sparsity in $\R$ is related to the
number of rotators.  A follow-on question is how these then relate to
sparsity in $\K$ and hence the sparsity of the implied graph.  This is
trivially explored by simulation and then simply counting the number
of zeros in $\R$ and $\K$. For a given set of rotator pairs $\cZ$  with
$z=|\cZ|,$  randomly pick which rotations will be non-zero then sample their
angles uniformly and generate $\R$ and $\K$. We repeat this process
10,000 times for each $0<z<m$.  For each dimension $q=20,30$,
Figure~\ref{fig:sparsity} shows the
median proportion of zeros in $\R$ and $\K$  as the
proportion of non-zero rotators increases. Note how quickly
the sparsity of $\K,$ defining the sparsity of the underlying graph,
decreases relative to $\R$.   This gives some insights into how the choice of the
prior sparsity probability $\beta_0$ plays a role in generating sparse graphs.

\begin{figure}[hbtp!]
\centering
\begin{tabular}{rl}
\hspace{-0.4cm}
\includegraphics[width=0.5\textwidth]{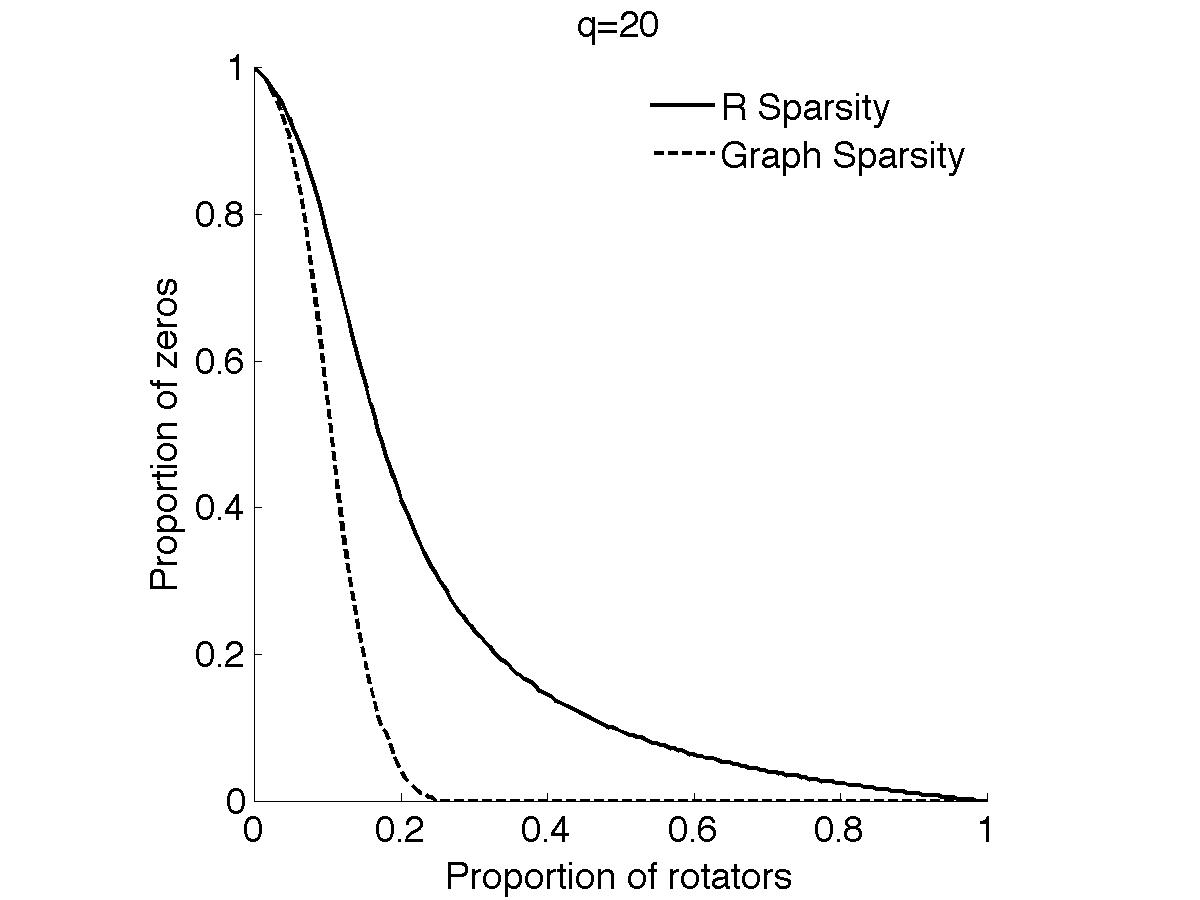} &
\hspace{-1.5cm}
\includegraphics[width=0.5\textwidth]{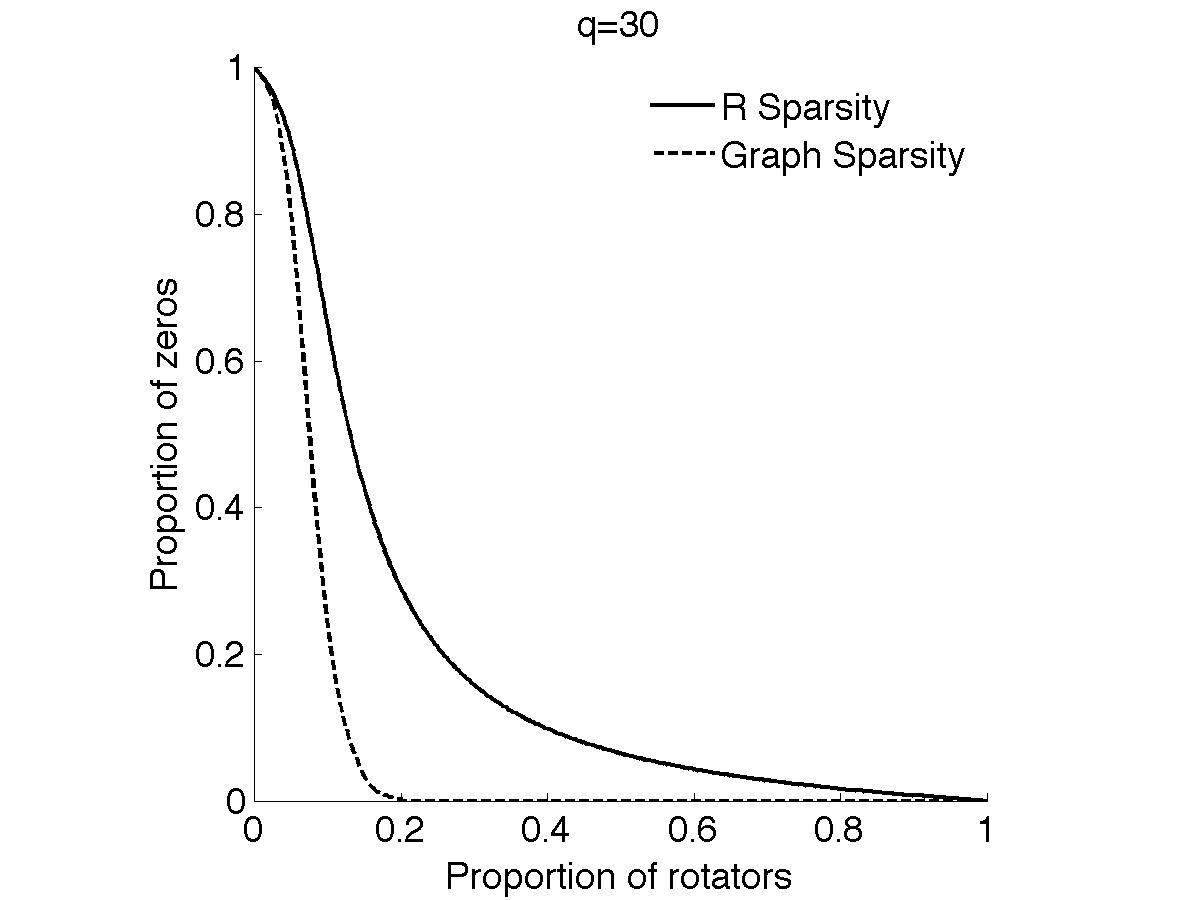}\\
\end{tabular}
\caption{Sparsity of $\R$ and the graphical model represented by $\K$ conditional
  on the number of random rotators in the model.  \index{sparsity prior} For $q=20$ and 30, this displays
  the prior median percent sparsity, i.e., proportion of zeros  in $\R$ and
  $\K$ out of the possible $q(q-1)$ and $q(q-1)/2$ respectively.  The prior for
 location of rotator pairs and values of the non-zero angles are each uniform, conditional
 on a given number of rotators selected.
  \label{fig:sparsity}}
\end{figure}

\section{Likelihood and Exploratory Analyses
\label{sec:explore}}

We discuss aspects of the likelihood function for the new parametrization $(\bomega,\A)$ without
considering sparsity, and then link that to an easily implemented forward selection algorithm that we have
found of use in defining starting values for full MCMC analysis under our sparsity priors.

\subsection{Likelihood\label{sec:like}}

Consider  a random sample  $\X=\{\x_1,\ldots, \x_n\}$ where
$ \x_i \sim N_q(\0,\V)$. With sum-of-squares
matrix $\S$, the log likelihood function has a form in $\R$ that is a
constant minus $\textrm{tr}(\R\A\R'\S)/2.$ Note that
$\O_{i,j}(\omega_{i,j})$ can be mapped onto an underlying $2\times 2$
rotation matrix $\G(\omega)$ where, for any $\omega,$
\begin{equation} \label{eq:ghkh}
\O_{i,j}(\omega) = \I_q+\H_{i,j}'(\G(\omega)-\I_2)\H_{i,j}\quad\textrm{with}\quad
\G(\omega) = \begin{pmatrix} \phantom{-}\cos(\omega) & \sin(\omega)\\
					   -\sin(\omega) & \cos(\omega)\end{pmatrix}
\end{equation}
and
\begin{equation} \H_{i,j}  \ = \
\bordermatrix{ 		&    &           &   &  i  &   &           &   & j  &    &          & \cr
                             & 0 & \cdots & 0 & 1 & 0 & \cdots & 0 & 0 & 0 & \cdots & 0 \cr
		          	& 0 & \cdots & 0 & 0 & 0 & \cdots & 0 & 1 & 0 & \cdots & 0
			}.
\end{equation}
Write $\R=\R_{i,j,0} \O_{i,j}(\omega_{i,j}) \R_{i,j,1}$; that is, $ \R_{i,j,0}$ is the
product of ordered rotators preceding $\O_{i,j}(\omega_{i,j}),$ and  $\R_{i,j,1}$ that following.
Also, define $\S_{i,j}=\R_{i,j,0}'\S\R_{i,j,0}$ and $\A_{i,j}=\R_{i,j,1}\A\R_{i,j,1}'.$
Note that  $\S_{i,j}$ has the interpretation of a decorrelated version of $\S$ based on the subset of rotators
represented in $\R_{i,j,0},$ i.e., all those preceding $(i,j)$ in the product making up $\R.$
Then, as a function of $\omega_{i,j}$ conditional on all other parameters, the log likelihood
reduces to
\begin{equation} \label{eq:omegaloglike}
\textrm{log}\
  p(\X|\omega_{i,j},-) = \textrm{c} -\textrm{tr}[
  \bPsi_{i,j}\G(\omega_{i,j})'\bPhi_{i,j}\G(\omega_{i,j})
  +\G(\omega_{i,j})'\bGamma_{i,j} ]/2
\end{equation}
where
$\bPhi_{i,j} =\  \H_{i,j}\S_{i,j} \H_{i,j}',$
$ \bPsi_{i,j} =\  \H_{i,j}\A_{i,j} \H_{i,j}$ and $\bGamma_{i,j}=\ \H_{i,j}\A_{i,j} \S_{i,j} \H_{i,j}'-\bPsi_{i,j}\bPhi_{i,j}.$
Some specific points to note are as follows:
\begin{itemize}

\item The form of the conditional likelihood is the kernel
  of a matrix Bingham von-Mises Fisher distribution for\index{Bingham von-Mises Fisher distribution}
  $\G(\omega_{i,j})$~\cite{Hoff2009}, which suggests such
  distributions as conditionally conjugate priors.

\item As a function of the scalar angle $\omega_{i,j},$ it is
  trivially shown that the log likelihood is a quadratic form in
  $(\sin(\omega_{i,j}), \cos(\omega_{i,j}))$.  It is easy to
  numerically maximize this conditional log likelihood. As a result,
  iterative maximum likelihood estimates can be derived by
  sequentially maximizing the above conditional likelihood functions
  as we iterate over rotators $(i,j),$ coupled with conditional
  maximization over the eigenvalues.

\item In the special case of $\R_{i,j,1}=\I_q,$ i.e., when $(i,j)$ is
  the right-most rotator pair and $\A_{i,j}=\A,$ the conditional
  likelihood can be maximized analytically if the diagonal of $\A$ is
  not constrained to be ordered. The maximizing $\omega_{i,j}$ value
  satisfies $\textrm{tan}(2\hat\omega_{i,j}) =
  2s_{i,j}/(s_{j,j}-s_{i,i})$ where $s_{a,b}$ are the scalar entries
  of the \lq\lq decorrelated'' sample variance matrix $\S_{i,j}.$
  Given this value, including $\A$ in the conditional log likelihood
  maximization gives the following:
\begin{equation} \label{eq:MLE} \begin{split}
	&\hat{\omega}_{i,j}=\frac{1}{2}\textrm{arctan}\left(\frac{2s_{i,j}}{s_{j,j}-s_{i,i}}\right), \\
	 &n\hat{a}_i^{-1}=s_{i,i}\textrm{cos}^2(\hat{\omega}_{i,j})+s_{j,j}\textrm{sin}^2(\hat{\omega}_{i,j})+2s_{i,j}\textrm{cos}(\hat{\omega}_{i,j})\textrm{sin}(\hat{\omega}_{i,j}),\\
	 &n\hat{a}_j^{-1}=s_{j,j}\textrm{cos}^2(\hat{\omega}_{i,j})+s_{i,i}\textrm{sin}^2(\hat{\omega}_{i,j})-2s_{i,j}\textrm{cos}(\hat{\omega}_{i,j})\textrm{sin}(\hat{\omega}_{i,j}),\\
	&n\hat{a}_k^{-1}=s_{k,k}\textrm{ for }k\neq i,j.
\end{split} \end{equation}

\item Continuing in the above case, if all correlation between
  variables $i$ and $j$ has been rotated away by the application of
  the preceding rotators so that $s_{i,j}=0,$ then the conditional MLE
  of $\omega_{i.j}$ is zero. In this case, it can also be shown that
  the conditional likelihood function in $\omega_{i,j}$ is
  proportional to $\cos^2(\omega_{i,j}).$

\item The above confirms the role of a continuous prior $p_c(\cdot),$
  as in \eqno{ctsprior}, as a conditionally conjugate prior centered
  around the region of no residual correlation between the two
  variables.

\end{itemize}

\subsection{An Exploratory Analysis Algorithm}\label{sec:explore_alg}

The investigations of likelihood structure above suggests a simple
exploratory analysis that can be of use in generating insights into
potential sparsity structure as well as, particularly, defining
starting values for a full Bayesian MCMC-based analysis of the sparse
eigenmatrix model.  This is discussed here in the case of $\kappa=0$
in the prior, for simplicity, although could be trivially modified.

Begin with $\S_*=\S$ and $\R_*=\I_q$ defining the \lq\lq current''
versions of the decorrelated sample variance matrix and corresponding
candidate eigenmatrix, respectively.

\begin{enumerate}
\item Set variable index $i=1$ and $j=2.$
\item Compute the sample correlation $r_{i,j}$ from $\S_*$.
\item If $|r_{i,j}|$ is large enough based on some pre-specified
  threshold which can be linked to the log likelihood difference it
  implies, add a rotator on pair $(i,j)$
  with angle $\hat{\omega}_{i,j}$ from \eqno{MLE}.
\item If the choice is to include a new rotator, update $\R_*$ to
  $\R_*\O(\hat\omega_{i,j}) $ and further decorrelate the sample
  variance matrix by updating $\S_*$ to
  $\O(\hat\omega_{i,j})'\S_*\O(\hat\omega_{i,j}).$
\item Sequence through the remaining $(i,j)$ in the order of the
  rotations in \eqno{decomp}.  Finally, set $\A_*$ to be the MLE based
  on the decorrelated $\S_*$.
\end{enumerate}
This forward selection process successively adds optimized rotators
via right multiplication, building up the corresponding sequence of
pairs of variable indices $(i,j)$ to define an empirical set $\cZ$ of
included rotators. It delivers this empirical estimate of $\cZ$ and
the corresponding, optimized estimate of $\R$, and hence of $\V$ and
$\K$ based on a final re-maximization of the likelihood for $\A$ given
the optimized $\R.$ There are many rules that can be used for the
thresholding in step 3. For instance, we could use the
absolute value of the conditional MLE, $\hat{\omega}_{i,j}$. However,
the effect of the rotation on the likelihood is unclear as it depends
on the eigenvalues which makes a particular threshold hard to
interpret. Simple thresholding on the current, \lq\lq residual''
sample correlation $|r_{i,j}|$ is natural and
interpretable; the squared correlation is the degree of residual
structure in $S_*$ that will be removed in that iteration of the
exploratory algorithm.   Hence a natural approach is to add a new rotator so
long as $|r_{i,j}|>\rho$ for some specified threshold $\rho.$

We note that this fast exploratory algorithm cannot constrain the
diagonal of our estimate $\A_*$ to be ordered. For exploration
purposes, this is not an issue as the resulting estimate of $\V$ and
$\K$ will not be affected. Furthermore if we denote by $\A_*$ an
exploratory estimate, and let $\P$ be the permutation matrix such that
$\A=\P\A_*\P'$ where the diagonal of $\A$ is ordered, then
$\V=\R\A\R'$ where $\R=\R_*\P$. In words, we are simply finding a
sparse spectral decomposition that has the interpretation of a forward
selection process based on residual correlation. We then order the
eigenvalues and their corresponding eigenvectors. Generating MCMC
starting values for $\bomega$ is now simply a matter of finding the
unique $\bomega$ that represents $\R$ in the general Givens rotator
product representation.  \cite{Anderson1987} recursively derive
$\bomega$ exactly from $\R$.  Just as many zero elements in $\bomega$
induce zeros in $\R$, the starting value for $\bomega$ based on the
decomposition of the sparse $\R$ will have many elements set to zero
making this a very fast and effective method for finding sparse
starting values.

In supporting material, we provide code implementing
this overall algorithm for interested readers.

\subsection{A Breast Cancer Genomics Example} 

\index{breast cancer genomics}\index{gene expression data}

We consider a subset of the microarray-based gene
expression data that is analyzed in more detail in Section~\ref{sec:mixtures} below. The subset of
size $n=66$ represents tumors that would be regarded as aggressive in
term of their expression profiles, based on higher levels of
expression of genes related to the two key hormonal pathway: the ER
(estrogen receptor) and Her2 growth factor
pathways~\cite{West2001,Huang2002}.  Activity of genes in these two
primary, distinct pathways, and their interactions with multiple other
biological pathways in cell growth and development, play into our
understanding of the heterogeneity of breast cancer, and critically
into advances in understanding clinically relevant cancer
subtypes~\cite{sorlie2004,Carvalho2008,Lucas2009a}.  Several of the
$q=20$ genes, notably the leading 6 in Figure~\ref{fig:exp_demo_corr}
(CA12, GATA3, HNF-3$\alpha$, LIV-1, Annexin, TFF3), are in part
co-regulated in the ER network, some being directly transcribed by ER
along with other factors.  These genes vary across this subset of
samples and are at relatively high levels of expression. These genes,
as well as other breast cancer biomarker genes (C-MYB, BCL-2) that
also interact with the ER network, play roles in multiple biological
pathways; as a result, their inter-relationships in expression are
more complicated than a simple one-dimensional ER factor would
explain.  Three of the variables (ERB-B2, HER2a, HER2b) are highly
related read-outs of activity of the hormonal Her2 pathway (the first
two are in fact different sequences from the same primary Her2 gene);
other genes in the sample (GRB7, CAB1) are known to be regulated or
co-regulated with Her2.  Two additional gene sequences (BRCA1, BRCA2)
relate to inherent susceptibility to breast cancer; their
transcriptional relationships with ER and Her2 are poorly understood,
although higher levels tend to be related to low ER and HER1 activity.
To give a sense of robustness, 5 additional variables are included:
the Junk genes represent random Gaussian noise.

\begin{figure}[bp!]
\centering
\begin{tabular}{cc}
\hspace{-1.0cm}$\rho=0.5$ & \hspace{-1.8cm}$\rho=0.25$ \\
\hspace{-1.0cm}\includegraphics[width=0.5\textwidth]{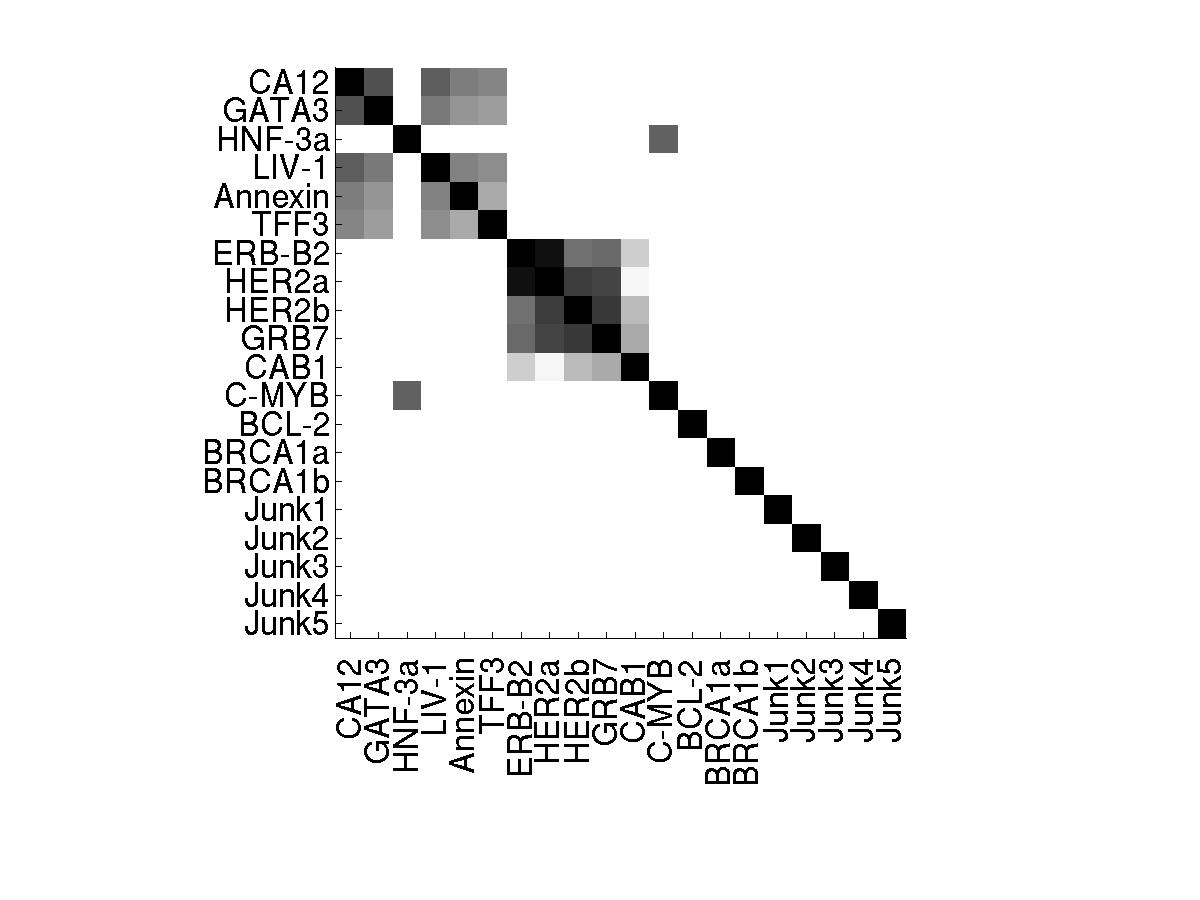} &
\hspace{-2.0cm}\includegraphics[width=0.5\textwidth]{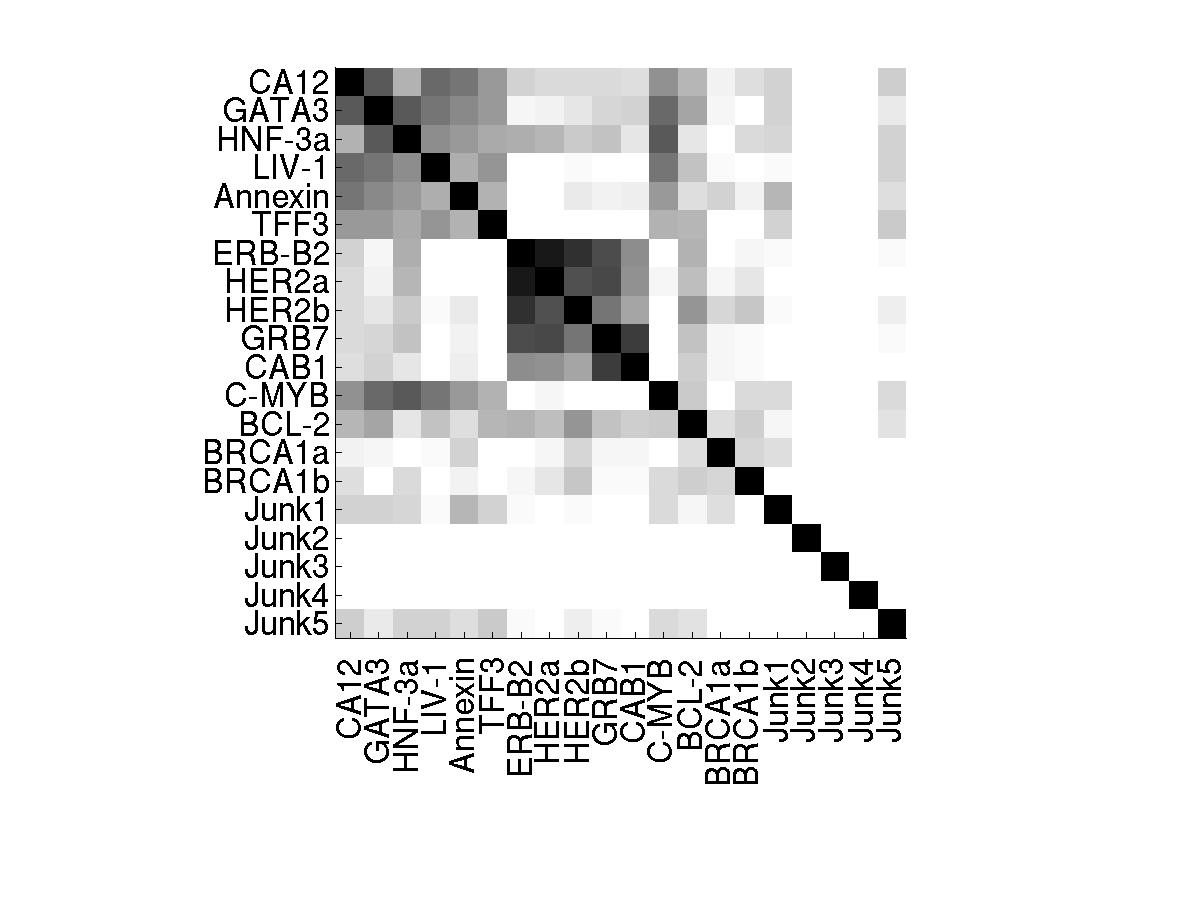} \\
\multicolumn{2}{c}{Sample correlations} \\
\multicolumn{2}{c}{\includegraphics[width=0.5\textwidth]{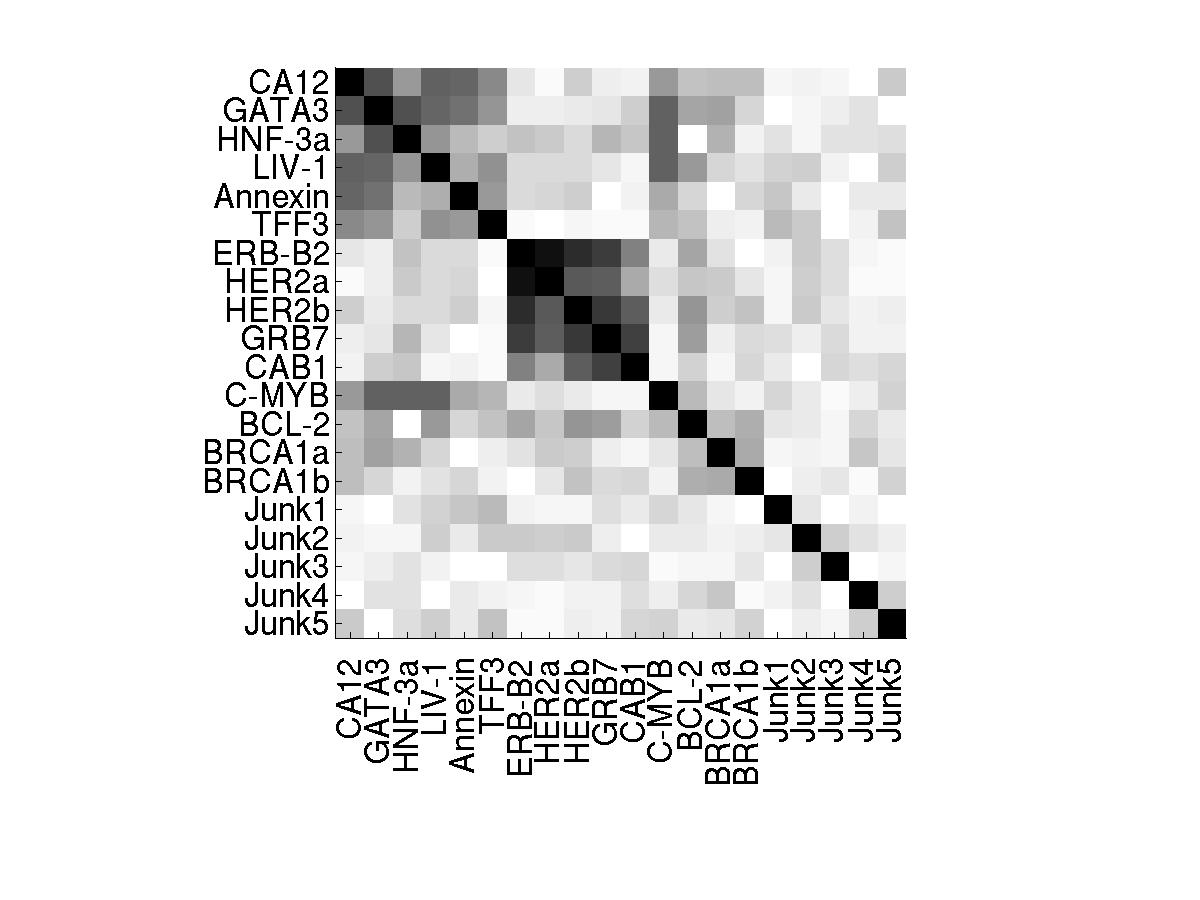}}
\end{tabular}
\caption{Breast cancer gene expression data example in Section~\ref{sec:explore_alg}.
Grey-scale heat maps show absolute values of elements in the  estimated
correlation matrix using the exploratory algorithm (white$=0$ correlation, black$=1$).
The first two frames correspond to stopping the algorithm when the maximum absolute value of the
 residual correlation drops below $\rho=0.5$ compared to $\rho=0.25,$ respectively. The third frame
shows the full sample  correlation matrix.
\label{fig:exp_demo_corr}}
\end{figure}

\begin{figure}[bp!]
\centering
\begin{tabular}{cc}
\hspace{-.1cm}$\rho=0.5$ & \hspace{-0.4cm}$\rho=0.25$ \\
\hspace{-.1cm}\includegraphics[width=0.4\textwidth]{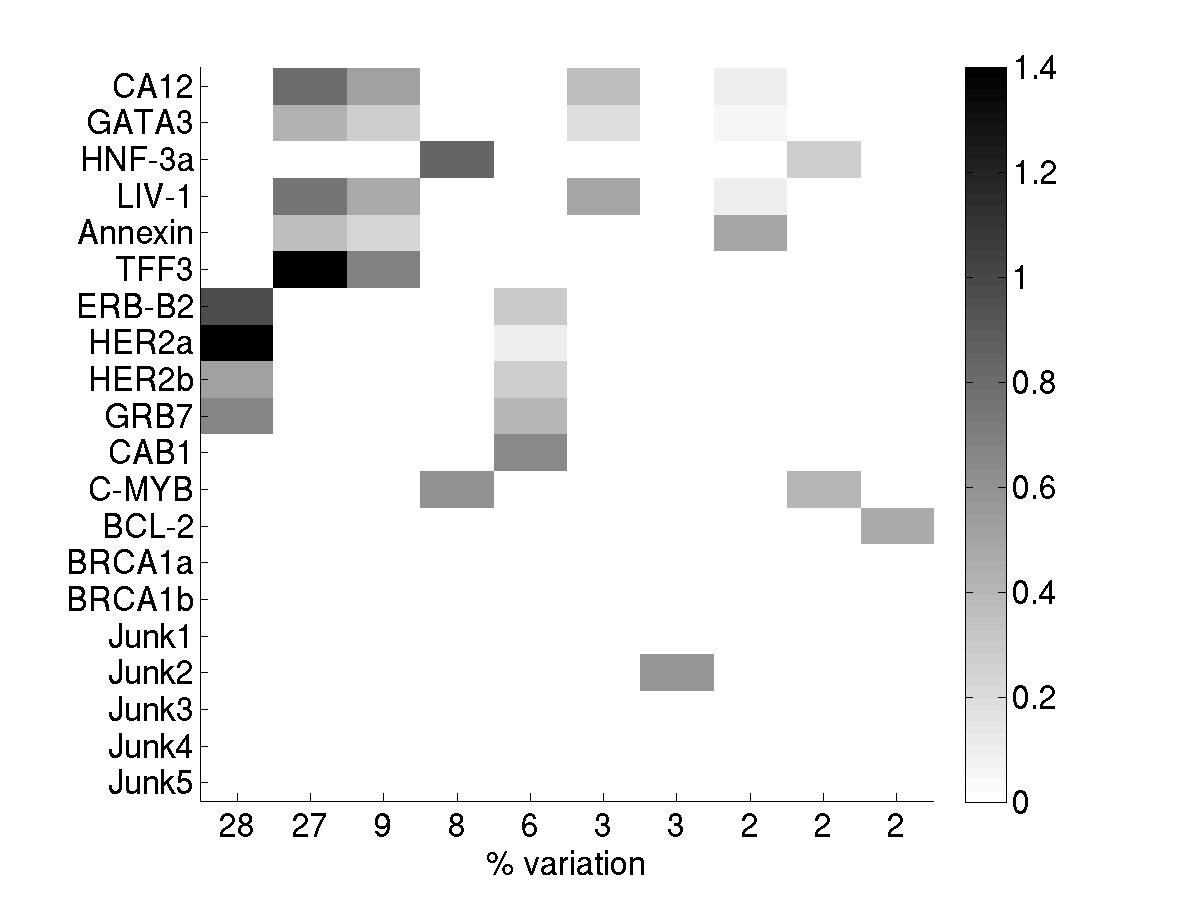} &
\hspace{-0.4cm}\includegraphics[width=0.4\textwidth]{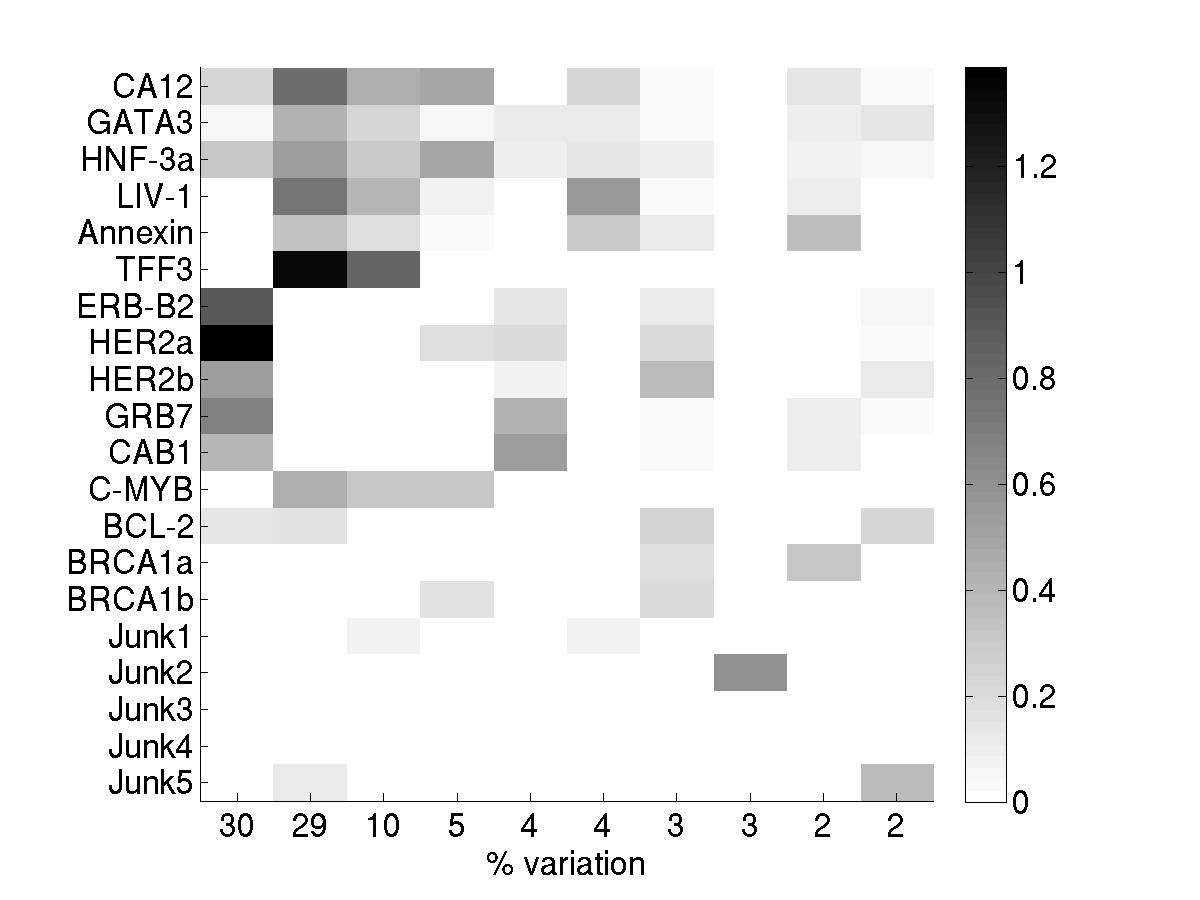} \\
\phantom{0}\\
\multicolumn{2}{c}{Sample eigenstructure} \\
\multicolumn{2}{c}{\includegraphics[width=0.4\textwidth]{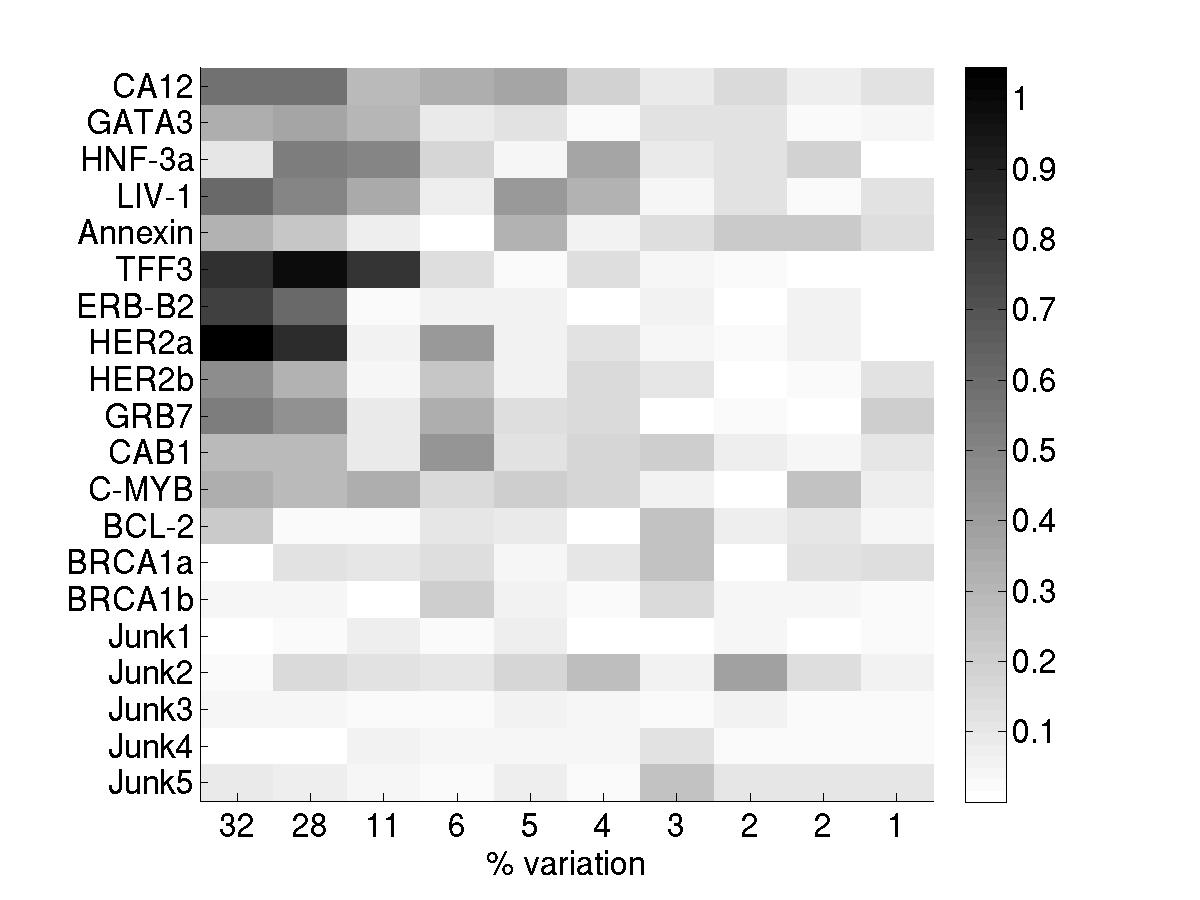}}
\end{tabular}
\caption{Breast cancer gene expression data example in Section~\ref{sec:explore_alg}.
Grey-scale heat maps show absolute values of elements in the first ten columns of the estimated  scaled eigenmatrix
$\R\D^{1/2}$  corresponding to the correlation matrix in Figure~\ref{fig:exp_demo_corr}. \index{sparse eigenmatrix}
The x-axes show the percent variation explained
by each eigenvector (column) as implied by the empirical estimates of the $d_j$ in each case.
\label{fig:exp_demo_eig}}
\end{figure}

After centering each of the $q=20$ variables, the exploratory analysis
was applied twice for a comparison of choice of stopping rule: we used
thresholds of $|r_{i,j}|\ge 0.5$ and $|r_{i,j}| \ge 0.25.$ A third
analysis simply computes the sample correlation matrix and the
corresponding eigenmatrix. Graphical summaries of the final estimates of
the  correlation matrix and the corresponding scaled eigenmatrix appear
in Figures~\ref{fig:exp_demo_corr} and \ref{fig:exp_demo_eig}, respectively.
We can see the increase in sparsity in
moving from no thresholding (the sample eigenmatrix) to a threshold of
0.25 and then 0.5, and how the sparse Givens construction-- via this
simple exploratory estimation method-- naturally denoises the raw
sample estimates.  The major ER and Her2 \lq\lq clusters'' evident in
the correlation matrices are sustained as we move up through the
levels of thresholding, and the corresponding \lq\lq factor loadings''
structure represented in the eigenmatrices successively reduces the
numbers and patterns of genes related to each factor (column).  The
most sparse structure in the first row shows that-- assuming this
level of sparsity-- we uncover a dominant Her2 factor loaded on four
of the Her2 cluster of genes, two main ER factors, and a few minor
factors that each represent only modest levels of variation explained
while contributing to the break-down of the complexity of expression
relationships in the data.

\section{Bayesian Analysis and Computation
 \label{sec:mcmc}}
We discuss and develop Bayesian computation for model fitting and exploration,
presenting customized MCMC methods.

\subsection{Overview}
In target applications with modest and increasingly high values of
$q,$ and hence larger $m=q(q-1)/2,$ the focus is on sparse structures
so that posterior distributions will concentrate on smaller numbers of
non-zero angles.  In these circumstances, visiting every element of
$\bomega$ using a Gibbs sampling approach will be computationally
expensive and other MCMC strategies are recommended.  \index{Markov chain Monte Carlo (MCMC)} 
Most effective
MCMC analysis can be achieved using reversible jump Markov Chain Monte
Carlo (RJ-MCMC) \cite{Green1995}. \index{reversible jump MCMC} 
We have implemented such an
approach based on exploring the space of non-zero elements of
$\bomega$ using a birth/death RJ-MCMC.  Each move through this \lq\lq
model space'' involves proposed changes that introduce non-zero
values, including the possible values $\omega=\pi/2,$ and/or setting
current non-zero values to zero.  We present the details of the
sampling algorithm by first outlining an approximation to the
conditional posterior of a single, non-zero $\omega_{i,j}$ that we
recommend as conditional proposal distribution for the MCMC.

\subsection{Wrapped Cauchy Proposals}\label{sec:cauchy}

Consider any rotator pair $(i,j)$ assuming $\omega_{i,j}\ne 0.$ The
conditional posterior is proportional to the conditional mixture
prior, mixing a point mass at $\pi/2$ with $p_c(\cdot),$ multiplied by
the conditional likelihood of the form discussed in the previous
section-- the conditional likelihood for $\omega_{i,j}$ given all
other rotators and eigenvalues.  Our MCMC adopts a conditional
proposal distribution for $\omega_{i,j}$ based on direct
approximation. Specifically, we use a proposal with pdf\index{wrapped Cauchy distribution}
$$ g(\omega) \equiv
\beta_{\frac{\pi}{2}} I(\omega=\pi/2) +
(1-\beta_{\frac{\pi}{2}})g_c(\omega)$$   where the continuous
density $g_c(\cdot) $  is that of a wrapped Cauchy
chosen to approximate the conditional posterior for
$\omega_{i,j}$ conditional on $0<|\omega_{i,j}|<\pi/2;$ i.e., a Cauchy
``wrapped" onto the interval $(-\pi/2,\pi/2)$~\cite{Fisher1993}.
Specifically,
$$g_c(\omega) =  \frac{1}{\pi} \frac{\textrm{sinh}(2\sigma)}{\{\textrm{cosh}(2\sigma)-\textrm{cos}[2(x-\theta)]\}} $$
where $(\theta,\sigma)$ are chosen so that $g(\cdot)$ approximates the
conditional posterior under prior $p_c(\cdot). $ The proposed values
of $(\theta,\sigma)$ are based on direct numerical approximation. We
set $\theta$ as the exact conditional posterior mode; this is easily
evaluated numerically.  Under any conjugate prior over non-zero
values, which includes our recommended default prior $p_c(\cdot)$ in
\eqref{eq:ctsprior}, note that $\log p(\omega_{i,j}|\X,-)$ is a
quadratic form in $(\sin(\omega_{i,j}), \cos(\omega_{i,j}))$ on a
bounded domain and can be evaluated along with any number of
derivatives very quickly. Resulting numerical maximization is then
routine and extremely efficient.  At the solution $\theta,$ the
curvature generates a value for the scale $\sigma$ from
\[\frac{1}{\sigma^2} = -\left.\frac{ \partial^2}{ \partial \omega_{i,j}} \textrm{log}\ p(\omega_{i,j}|\X,-)\right|_{\omega_{i,j}=\theta}. \]
The wrapped Cauchy form can be viewed as a diffuse posterior
approximation-- the result of an initial Laplace approximation subject
to inflating the tails to ensure good coverage of the exact
conditional posterior. To deal with cases in which the mode $\theta$
lies on the boundary, simply replacing $g_c(\cdot) $ with a
$\textrm{Be}(\omega|0.25,0.25)$ density has been empirically found to
provides an effective, default proposal.

\subsection{Reversible Jump MCMC \label{sec:RJMCMC}}

Denote all  parameters of interest by $\bTheta=\{\cZ, \bomega, \A\}$
where, as introduced above,  $\cZ$ is the set of pairs of indices $(i,j)$
corresponding to included rotators with non-zero angles.  In an overall MCMC, suppose
we are at a current state at iterate $t$ with parameters
$\bTheta^{(t)}=\{\cZ^{(t)},\bomega^{(t)},\A^{(t)}\}$  with a current $z^{(t)}=|\cZ^{(t)}|$ rotators.
Consider now either adding or removing a rotator index pair
 from $\cZ^{(t)}$. Set probabilities of adding a rotator
(birth) and removing a rotator (death) at values denoted by $p_B$ and $p_D$, respectively.\index{reversible jump MCMC}
For a proposed birth, randomly select an ordered pair,
$(i^*,j^*) \in \cM\setminus\cZ^{(t)}$ to index a proposed angle,
and then generate a proposal $\omega^*_{i^*,j^*}$
from $g(\omega_{i^*,j^*})$ described in the Section~\ref{sec:cauchy}.
This implies the following birth step accept/rejection ratio:
\begin{equation} \label{eq:bdratio} \alpha_B = \frac{p(\X |
    \omega^*_{i^*,j^*},-)
    p(\omega^*_{i^*,j^*})(z^{(t)}+1)p_D}{p(\X|\omega^{(t)}_{i^*,j^*},-)
    p(0)g(\omega^*_{i^*,j^*})(m-z^{(t)})p_B}.
\end{equation}
For a proposed  death step,   choose an element from $\cZ^{(t)}$ uniformly and set its
corresponding angle to zero. The resulting rejection ratio is simply $\alpha_D=\alpha_B^{-1}.$
We then set $\omega^{(t+1)}_{i^*,j^*}=\omega^*_{i^*,j^*}$ with probability $\min(1,\alpha_B)$.
To facilitate better mixing, we do several reversible jump proposals in each MCMC iteration.
This results in the updated (possibly, of course, also the same) set of  rotator pairs
$\cZ^{(t+1)}.$

The MCMC next updates all
non-zero angles indexed in $\cZ^{(t+1)}$. For each pair
$(i,j)\in \cZ^{(t+1)}$,  generate a proposal $\omega^*_{i,j}$ from $g(\omega_{i,j})$
and set $\omega_{i,j}^{(t+1)}=\omega^*_{i,j}$ with probability
\begin{equation} \label{eq:mhratio} \min \left\{ 1,
  \frac{p(\X|\omega^*_{i,j},-)p(\omega^*_{i,j})g(\omega^{(t)}_{i,j})}
  {p(\X|\omega^{(t)}_{i,j},-)p(\omega^{(t)}_{i,j})g(\omega^*_{i,j})} \right\}.
\end{equation}
Since we update the elements of $\bomega^{(t)}$ in order, the recursive relationship
between the conditional log likelihoods given in \eqno{omegaloglike}
makes computation extremely fast. This recursive update is initialized at
$$
\A_{i_1,j_1}=\O_{i_1,j_1}(\omega^{(t)}_{i_1,j_1})'R^{(t)}\O_{i_1,j_1}(\omega^{(t)}_{i_1,j_1})\quad\textrm{and}\quad
\S_{i_1,j_1}=\S.
$$
The log likelihood can then be optimized and evaluated
quickly as described in the previous section. For $k>1$ as we step through updates of rotator $k,$  compute
\begin{equation*}
\begin{split}
\A_{i_k,j_k}&=\O_{i_k,j_k}(\omega^{(t)}_{i_k,j_k})'\A_{i_{k-1},j_{k-1}}\O_{i_k,j_k}(\omega^{(t)}_{i_k,j_k})\quad\textrm{and}  \\
\S_{i_k,j_k}&=\O_{i_{k-1},j_{k-1}}(\omega^{(t+1)}_{i_{k-1},j_{k-1}})'\S_{i_{k-1},j_{k-1}}\O_{i_{k-1},j_{k-1}}(\omega^{(t+1)}_{i_{k-1},j_{k-1}}).
\end{split}
\end{equation*}
Note that the recursions for $\A_{i_k,j_k}$ and $\S_{i_k,j_k}$ only involve
taking linear combinations of two rows and columns, so we can avoid
recomputing the entire eigenmatrix for every proposal. In words: (i)
start with the sum of squares matrix and the eigenmatrix; (ii)  remove
the first rotation by multiplying by its transpose; (iii)  perform a fast
Metropolis move by exploiting the quadratic form of the likelihood; (iv)
decorrelate the sum of squares matrix with the new rotator and
remove the next rotation from the eigenmatrix.

The final step is to update the entries of the diagonal precision matrix $\A$ from the
current value $\A^{(t)}. $ Let
$$\B = \R^{(t)\prime}\S\R^{(t)}.$$
Then for each $j=1,\dots,q$
\begin{equation}\label{eq:postaj}
p(a^{(t)}_j|\X,-) \propto \textrm{Ga}(a^{(t)}_j | (\eta_1+n)/2,(\eta_2 + B_{j,j})/2)
\ I(a^{(t)}_{j-1} < a^{(t)}_j < a^{(t)}_{j+1})
\end{equation}
where $Ga(x|a,b)$ denotes the pdf of the $Ga(a,b)$ distribution evaluated at $x.$  Since
the  eigenvalues are constrained to be ordered, the conditional
distributions are constrained as well. The resulting constrained gamma distribution is
sampled using  the inverse cdf method.

\section{Simulation Study and Comparisons\label{sec:simstudy}}

We make a detailed, simulation-based comparison of the modelling approach with analysis using
traditional Gaussian graphical modelling (GGM)
~\cite{lauritzen96,Dobra2004,Jones2005,carvalho:west:07,Dobra2011,Rodriguez2011}.
The GGM framework with decomposable graphs is directly comparable and stands as a
current  benchmark model  context.\index{Gaussian graphical model}\index{decomposable graphical model} 

The simulation study was conducted using zero mean normal models in each of
$p\in\{10, 20, 30,40,50,75,100\}$ dimensions, with a fixed sample size of $n=150$ observations.
Synthetic data sets generated from specific
model classes were analyzed using the sparse Givens approach and the GGM approach, the
latter using shotgun stochastic model search for the Bayesian analysis~\cite{Jones2005}.
Each analysis was repeated for 100 simulation samples.  The underlying models and
synthetic data generation proceeded as follows:
\begin{enumerate}
\item Generate a target precision matrix  $\K=\U'\U$ where $\U$ is upper triangular with:
 \begin{itemize}
    \item $\U_{i,i}=\sqrt{\nu_i}$ with $\nu_i\sim\chi^2_{p-i}$, $(i=1,\ldots,p), $
    \item $\U_{i,j}=u_{i,j}I(|u|_{i,j}>1)$ where $u_{i,j}\sim N(0,1)$,  $(j>i,\ i=2,\ldots,p-1), $ and
    \item the $\nu_i, u_{i,j}$ are mutually independent.
  \end{itemize}
\item Draw $n=150$ observations $\X$ as a random sample from $N(0,\K^{-1})$.
\item Fit the sparse Givens model using 15{,}000 MCMC iterations. Discard the first 10{,}000
and save the final 5{,}000 as a Monte Carlo sample for sparsity patterns and values of $\K.$
\item Fit the GGM model using 15{,}000 stochastic search iterations. Discard the first 10{,}000 and save the
final 5{,}000 graphs identified, their posterior probabilities and the parameters of the corresponding
posteriors for $\K$ on each graph.
\end{enumerate}
For prior distributions, the probability of including a \lq\lq free'' parameter
was set to $2/(p-1)$. In the sparse Givens models, a free parameter is an angle; in GGMs, it is
the probability of including a random edge. This specification aims to match the prior expectations of
degrees of sparsity between the two approaches.

Comparisons are based on measuring agreement between the approximate posteriors and the \lq\lq true''
underlying data-generating distribution.   For any precision matrix $\K_*$, we can directly compute the
Kullback-Leibler (KL) divergence \index{Kullback-Leibler divergence} 
of the $N(0,\K_*^{-1})$ distribution from the true, underlying
$N(0,\K^{-1})$ distribution. With both the MCMC posterior samples and the GGM search results,
we can then approximately evaluate the posterior distribution for the KL divergence of the
chosen model from the truth.

Figure~\ref{fig:simstudy} summarizes the posteriors for the KL divergences,
aggregated across 100 repeat samples. We can see that in 20 dimensions, both methods
perform similarly; this  is not surprising since there is reasonable amount of data
relative to the dimension. However, in 30 dimensions, the new sparse Givens approach is significantly
better, and its dominance is progressively more pronounced as the dimension increases.

\begin{figure}[htbp!]
\centering
\includegraphics[width=0.4\textwidth]{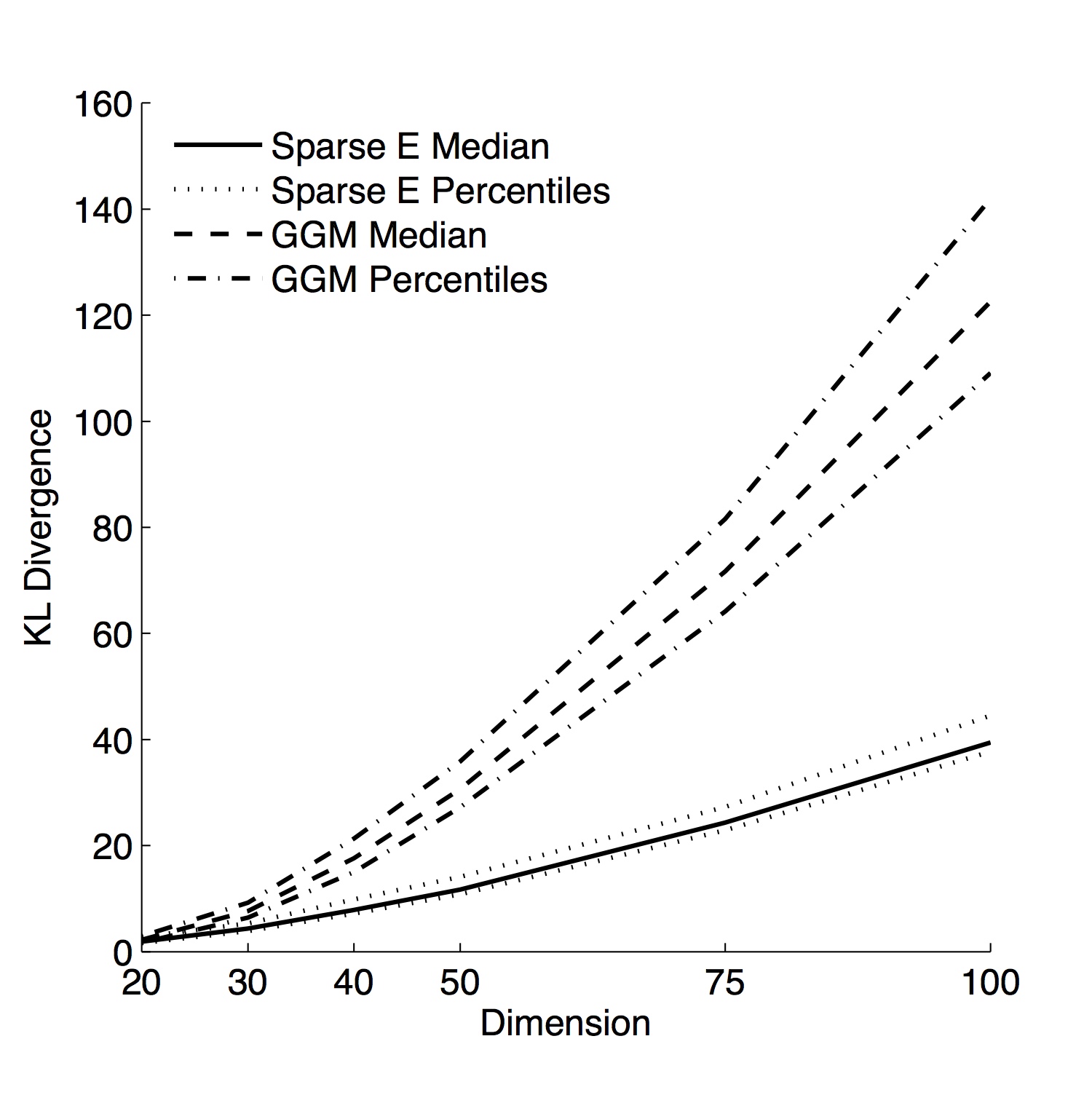}
\caption{\protect Simulation study comparing GGM to the sparse Givens (labelled \lq\lq Sparse E\rq\rq in the figure)
         methodology. Shown are the medial, 10th and 90th percentiles of the posterior for (logged values of)
          the KL divergence from the fitted model to the true underlying data generating model.
	\label{fig:simstudy}}
\end{figure}

\section{Mixtures of Sparse, Full-Rank Factor Models \label{sec:mixtures}}

\index{mixtures of factor models} \index{mixtures of normals} \index{Bayesian factor analysis} 

Many applied contexts involve evident non-Gaussian structure as well as measurement error overlaid on the underlying
dependency patterns we are interested in inferring via the sparse eigenmatrix models.
The gene expression example of Section~\ref{sec:explore_alg} is just one case-in-point.  First,
in the broader contexts of samples from the full breast cancer population, there is inherent non-Gaussianity
representing heterogeneity in cancer states. This heterogeneity can be regarded as arising from a
mixture of sub-populations, or \lq\lq subtypes'' that, in terms of expression data outcomes,
are hugely overlapping~\cite{sorlie2004,Carvalho2008,Lucas2009a}.
More broadly, use of discrete mixtures of Gaussians is a well-established
strategy for modelling what might be quite non-Gaussian distributional forms, whether or not there is an   inherent in mixture components and discrimination/classification~\cite{Escobar1995,Chan2008,Suchard2010}.  Second, measurement errors are ubiquitous. Again the
gene expression example and broader context is a good example, as the experimental and
data extraction contexts are well-known to overlay underlying biological variation with meaningful
uncertainties that must be accommodated within a more general model in order to avoid obscuring
relationships and leading to potential biases in resulting inferences~\cite{Lucas2006,Carvalho2008,Lucas2009,Lucas2010}.

\subsection{Mixture Models and Extension of MCMC Analysis}

We address the above, general considerations with Gaussian mixture models overlaid with measurement errors. Each mixture component has a variance matrix modelled via
the sparse Givens strategy; this can be directly interpreted as a sparse,
 full-rank latent factor model for underlying `\lq\lq structural''
dependencies. As a result, the overall framework is a generalized, adaptively sparse model  for
 \lq\lq mixtures of (full-rank, sparse) factor analyzers''~\cite{McLachlan2003,Rodriguez2011}.

Assume we observe $n$ independent $q-$vector observations
$\Y = \{ \y_1, \dots, \y_n \} $ where $ \y_i = \x_i + \bepsilon_i $ with independent measurement errors
$ \bepsilon_i \sim N(\0, \bPsi)$ having variance matrix $\bPsi=\textrm{diag}(\psi_1,\dots,\psi_q). $
Suppose the latent signals $\x_i$ are independently  drawn from a discrete mixture of multivariate normals
having pdf
$$ p(\x) = \sum_{c=1}^C w_cN(\x|\bmu_c,\bSigma_c).$$
Equivalently,
$$\y_i|\gamma_i \sim N(\bmu_c, \bSigma_c + \bPsi), \quad\textrm{where}\quad Pr(\gamma_i=c)=w_c,$$
involving the underlying latent  mixture component indicators $\gamma_i$
that are independently drawn from the multinomial distribution on cells $\{ 1:C\}$ with the vector of cell
probabilities  $\w=(w_1,\ldots,w_C)'.$

We develop this mixture model under sparse Givens factor structures for each of the
mixture components. That is,  $\bSigma_c=\R_c\D_c\R_c'$
where we model each of the $(\R_c,\D_C)$ with the prior structure of
Section~\ref{sec:priors}, independently across components $c$.  This allows for
differing degrees and patterns of sparsity as we move across components of the mixture,
in the context of also accommodating realistic assessment of overlaid measurement errors.
We couple this with conditionally conjugate normal priors for the $(\bmu_c|\bSigma_c)$
independently across components, and a
similarly conditionally conjugate Dirichlet prior for the mixture weights  $\w$. The final component of prior
specification is a set of $q$ conditionally independent inverse gamma priors for the measurement error
variances $\psi_j,$ $j=1:q.$

The traditional MCMC analysis of multivariate normal mixtures~\cite{Lavine1992,Cron2011} is easily extended to apply here.
Several points require note. At each iterate conditional on currently imputed values of the
$\x_i,$ we resample new values of the
mixture component indicators $\gamma_i$ for each of the $i=1:n$ observations. Conditional on these
indicators, the imputed signal \lq\lq data''  vectors
$\x_i$ are organized into $C$ conditionally independent normal subgroups.
The inherent component labelling issue is automatically addressed each iteration
using the efficient component relabelling strategy of~\cite{Cron2011}.  The numbers of
observations in each group define the conditional multinomial sample needed to draw new values of
the mixture weights $\w$ from the implied conditional Dirichlet posterior.
We then resample new group means $\bmu_c$ from the implied set of $C$ conditional normal posteriors.
Subtracting the group means from the values of the $\x_i$ within each group, we are then in a context of
having $C$ replicates of the normal, sparse Givens model.  Hence we apply $C$ parallel RJ-MCMC steps
to draw new values of rotator index sets, angles and eigenvalues in each of the components.  The final, additional
component of the overall MCMC posterior simulation arises due to the additive measurement error
structure of the model.   Given the resampled parameters and component indicators, the
implied conditional posterior for each $\x_i$ is normal, so easily sampled; given the new value of $\x_i,$
we compute new synthetic residuals $ \bepsilon_i = \y_i -\x_i $  that lead to a set of $q$,  independent
conditional posteriors for the $\psi_i$ that are each of inverse gamma form.

Additional technical details of these
steps appear in the Appendix, and in supporting material, we provide code implementing this
MCMC algorithm.

\subsection{A Broader Study in Breast Cancer Genomics}
\index{breast cancer genomics}\index{gene expression data} 

We analyse a set of $n=295$ breast cancer gene expression sample that represent the full range of breast cancers;
the data set uses the same 15 genes  as in the example of Section~\ref{sec:explore_alg}, but now
reflecting  full population heterogeneity; variations in the expression levels of these
15 breast cancer related genes is much greater across this full set of 295 tumor samples.
The $q=20-$dimensional data set again includes 5
 \lq\lq Junk'' genes generated as Gaussian noise, to add dimension for the evaluation of the model
analysis.

As discussed in Section~\ref{sec:explore_alg}, breast cancer heterogeneity based on molecular markers related to ER and Her2 pathways is often regarded in terms of over-lapping cancer subtypes.   The genes selected for this study relate to these pathways, and the
variability across samples is certainly empirically consistent with at least  three
underlying components; see some scatter plots on a few genes in Figure~\ref{fig:scatter}.
We fit the sparse Givens,  finite mixture model $C=4.$    Priors for residual measurement error variances
are informed by a wealth of prior information from studies of gene expression data using
Affymetrix microarrays in breast cancer and other contexts~\cite{Seo2004,Seo2007,Carvalho2008,Lucas2006,Lucas2009,Lucas2009a,Lucas2010}.
Specifically, we adopt
$\psi_j \sim InvGa(3.1, 0.17)$ with implied 95\% prior credible intervals for
measurement error standard deviations of about $(0.15, 0.5).$
For the mixture weights $\w,$ we take a uniform Dirichlet prior. For component
locations, we take $\bmu_c|\bSigma_c \sim N(\0,
\tau\bSigma_c)$ where $\tau$ is large to induce a rather diffuse marginal  prior
on  mixture  locations;  the analysis summarized below has $\tau=1{,}000.$
The prior over the sparse Givens parameters for each $\bSigma_c$ takes  $\beta_0=0.99$,
$\beta_{\pi/2}=0.25$, and $\kappa=0$. This prior  expresses an expectation of a fair
degree of sparsity in each  $\R_c,$  coupled with a vague uniform prior on values of
non-zero angles. Finally for $\D_c$, we have  $\eta_1=1/1{,}000$ and $\eta_2=1/1{,}000$
representing an uninformative prior on the eigenvalues, up to the constraints imposed by
their ordering.

 For starting values, we first crudely
partition the data using $k-$means clustering, then use the exploratory
algorithm of Section~\ref{sec:explore_alg}  with a correlation
threshold of $0.5.$  We  run the MCMC for 200{,}000 iterations, discarding the first 100{,}000
for burn-in to ensure convergence, with a number of subjective assessments of this.
The analysis identifies 3 main components with posterior means of components
weights of $(0.564, 0.303, 0.125, 0.008)'.$ Figure~\ref{fig:scatter}  shows
some aspects of the posterior through scatter plots of data on a few selected genes.
From the MCMC we compute estimates of the sample component classification probabilities
$Pr(\gamma_i=c|\Y)$ for each $i=1:n,$ and allocate sample $i$ to the most probable component
for the purpose of this graphical display; the data points are plotted as symbols corresponding to
their most probable component.  The dominant identified component $c=1$ represents cases
with expression varying across high levels for genes linked to the ER pathway,  including TFF3, CA-12 and GATA3 shown in Figure~\ref{fig:scatter}, and with Her2 pathway genes varying at relatively low levels; these
represent the broad luminal subtype of breast cancers~\cite{sorlie2004,Carvalho2008,Lucas2009a}.
The second main component $c=2$ represents cases generally high in Her2 expression levels, with
other genes varying across the spectrum; this corresponds to  high-risk  Her2 breast cancers that
are generally targets for the Her2 receptor inhibiting drug herceptin. The third, smaller component
$c=3$ represent the so-called triple-negative/basal-like tumors, with generally low levels of
activity of both ER and Her2 related genes. The patterns in the figure, and in those of other genes in the
example, as well as the posterior estimates of relative sizes of these three main components, are
quite consistent with the known cancer biology and relative probabilities of these three broad, and
imprecisely defined clinical subtypes of tumors.

\begin{figure}[tp!]
\centering
\begin{tabular}{rl}
\hspace{-1.6cm}\includegraphics[width=0.5\textwidth]{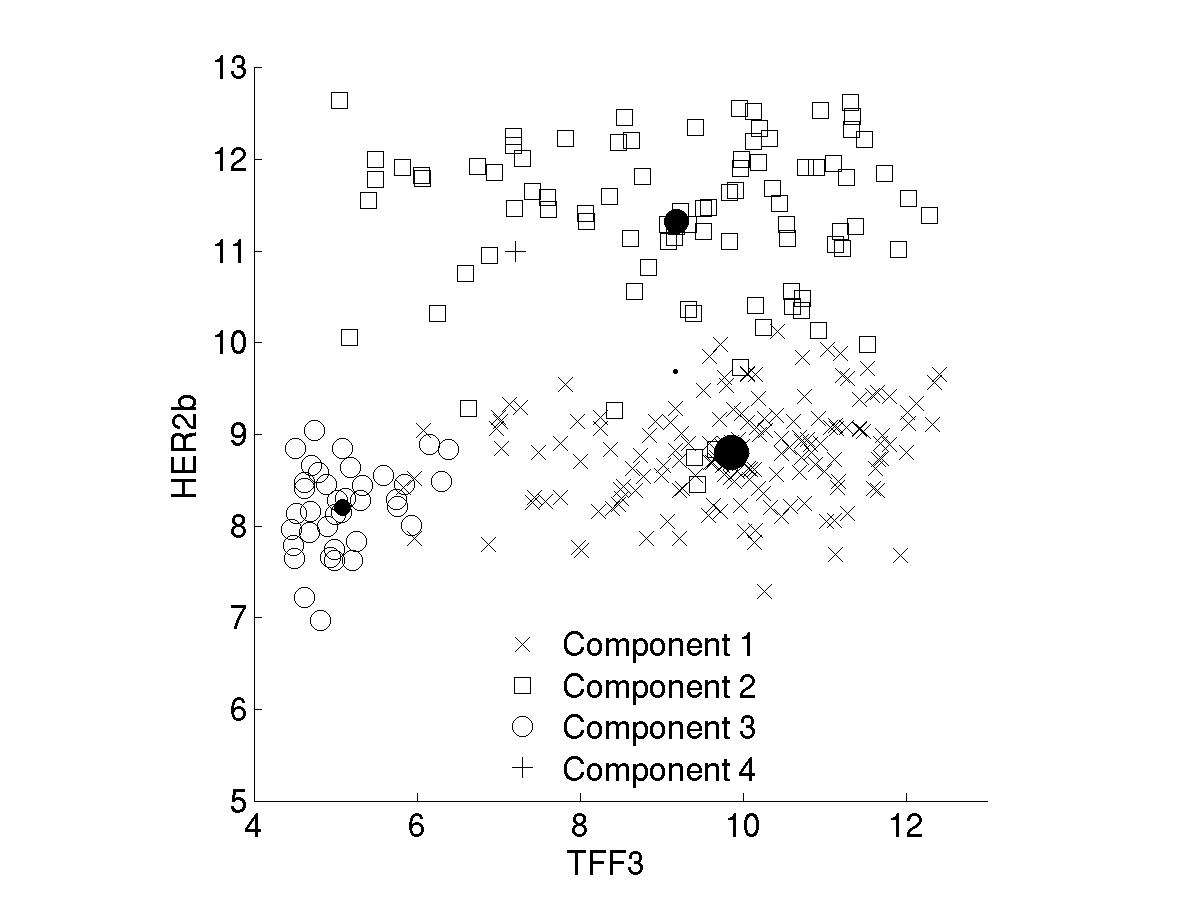} &
\hspace{-2.1cm}\includegraphics[width=0.5\textwidth]{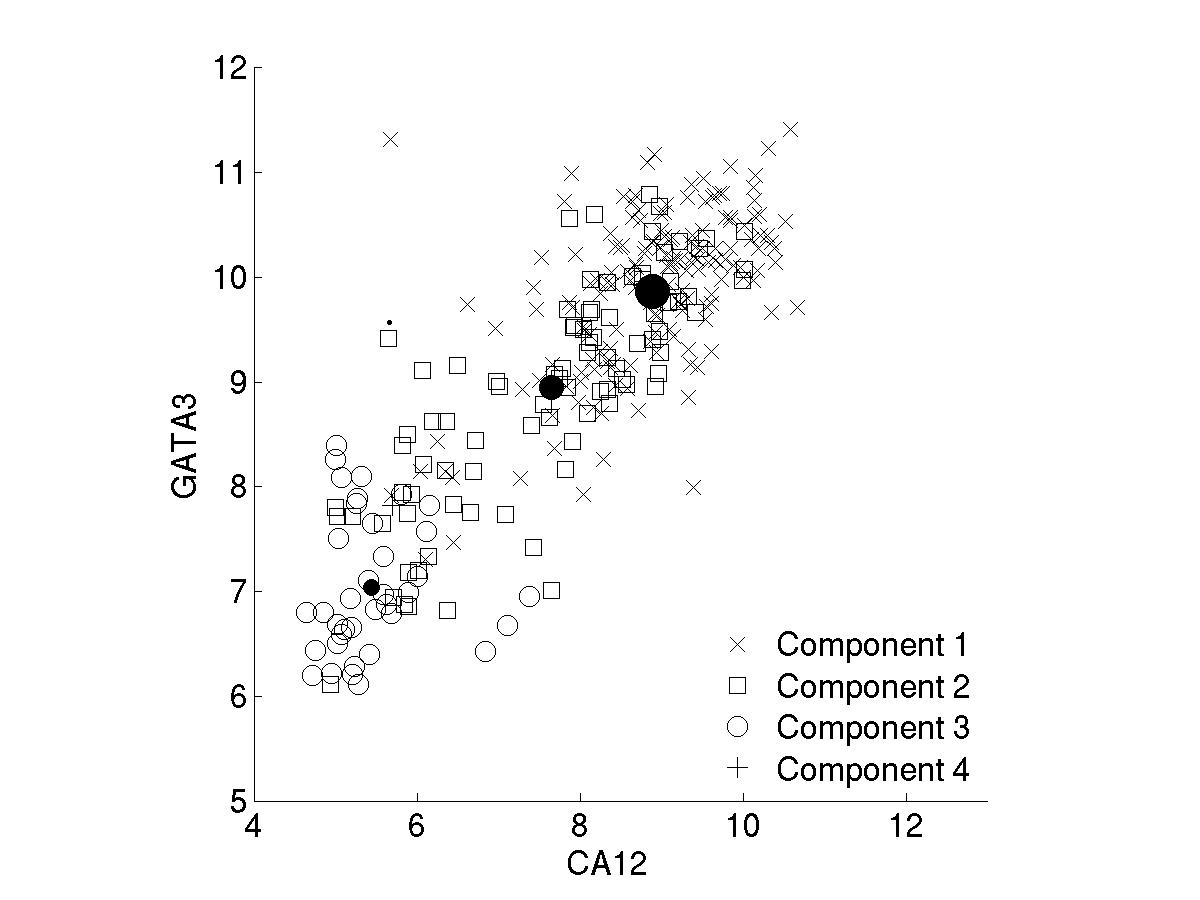}
\end{tabular}
\caption{\protect Aspects of estimated component classification summaries
	 in the analysis of the mixture model for the gene expression data.
	The filled  circles correspond to the estimated means of the mixture components
        and their sizes are proportional to the estimated component weights.
	\label{fig:scatter}}
\end{figure}

Table~\ref{tab:postsparse}
summarizes the posterior for the number of non-zero angles in each mixture
component and the sparsity of $\R_c$.  The
maximum number of rotators is 105 and the maximum sparsity of $\R$ is
210. The posterior favours a very sparse set of angles in
each mixture component and the eigenstructure in each component is
quite sparse as a result.

\begin{table}[bp!]
\centering
\begin{tabular}{lrrrrcrrr}
\multicolumn{1}{l}{}
 && \multicolumn{3}{c}{\% non-zero rotators} && \multicolumn{3}{c}{\% zeros in $\R$}\\
 &
 &\multicolumn{1}{c}{{2.5\%}}&\multicolumn{1}{c}{{50\%}} &\multicolumn{1}{c}{{97.5\%}}
 &
 &\multicolumn{1}{c}{{2.5\%}}&\multicolumn{1}{c}{{50\%}}&\multicolumn{1}{c}{{97.5\%}}\\
\cline{3-5} \cline{7-9} \\
{Component 1}&&9.5&10.5&12.1&\phantom{.}\qquad\qquad&66.8&71.8&73.7\\
{Component 2}&&8.4&10.0&11.0&&74.5&82.6&87.4\\
{Component 3}&&1.6&3.2&4.2&&95.3&92.1&97.6\\
\phantom{.}\\
\end{tabular}
\caption{Posterior medians and end-points of approximate posterior 95\% credible intervals for
the percentage of sparse elements of $\R_c$  and of the number of rotators in each of the
mixture components $c=1,2,3$.\label{tab:postsparse}}
\end{table}

Figures~\ref{fig:postKedges} and~\ref{fig:postEmean} give
graphical summaries generating insights into the inferred sparse structures
underlying the $\bSigma_c$ for each component $c=1:4.$ Figure~\ref{fig:postKedges}
shows heat maps of approximate posterior probabilities of non-zero values in the
precision matrices $\K_c=\bSigma_c^{-1}$ indicating the nature of sparsity and the
underlying graphical model structure.
Component $c=1$ has high probabilities on multiple edges linking pairs of ER related
genes, Her2 related genes and tying in the two BRCA genes.  Component $c=2$
more sharply identifies a Her2-related cluster and a distinct ER-related cluster,  with
somewhat weaker links to  the two related BRCA genes. The much sparser component $c=3$
highlights links only between Her2 related genes.

\begin{figure}[bp!]
\centering
\begin{tabular}{rl}
\hspace{-1.0cm}\includegraphics[width=0.5\textwidth]{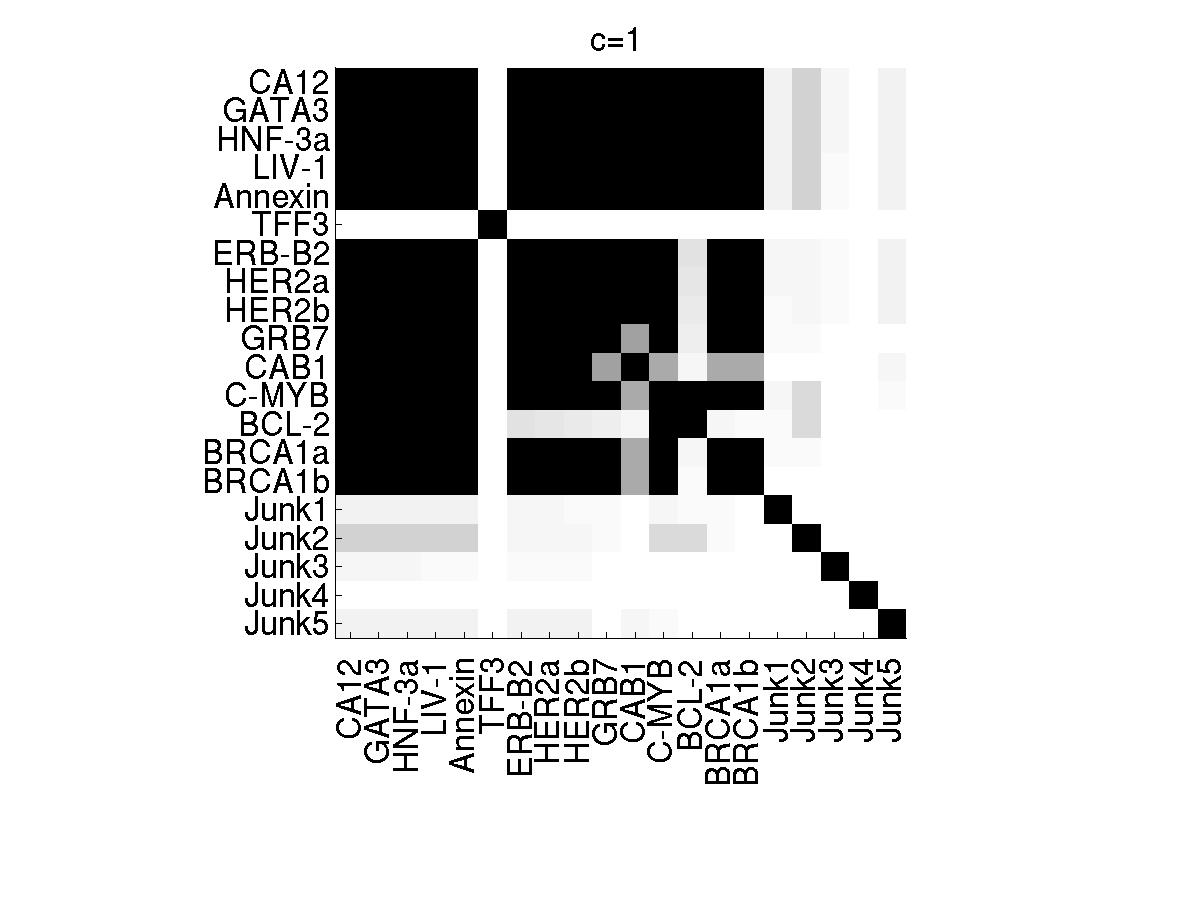} &
\hspace{-2.0cm}\includegraphics[width=0.5\textwidth]{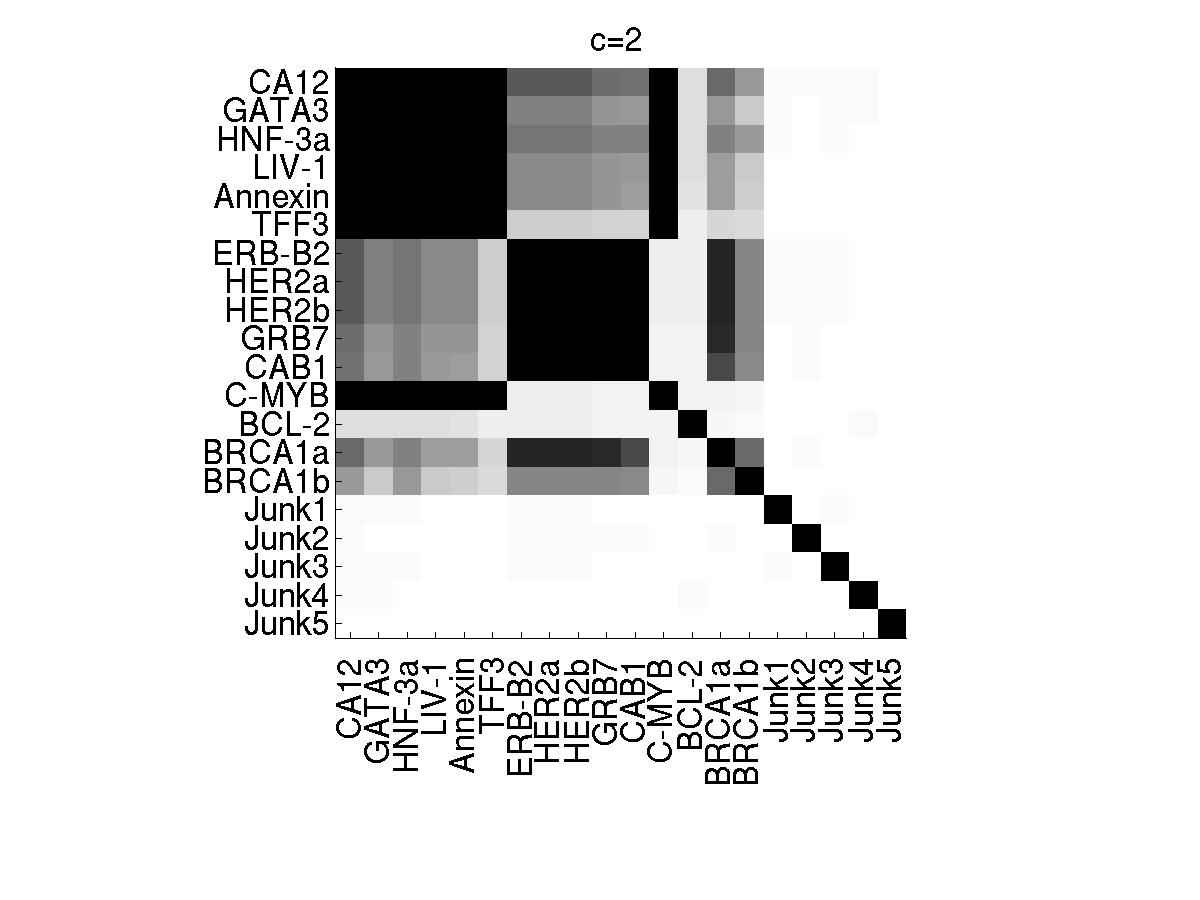} \\
\hspace{-1.0cm}\includegraphics[width=0.5\textwidth]{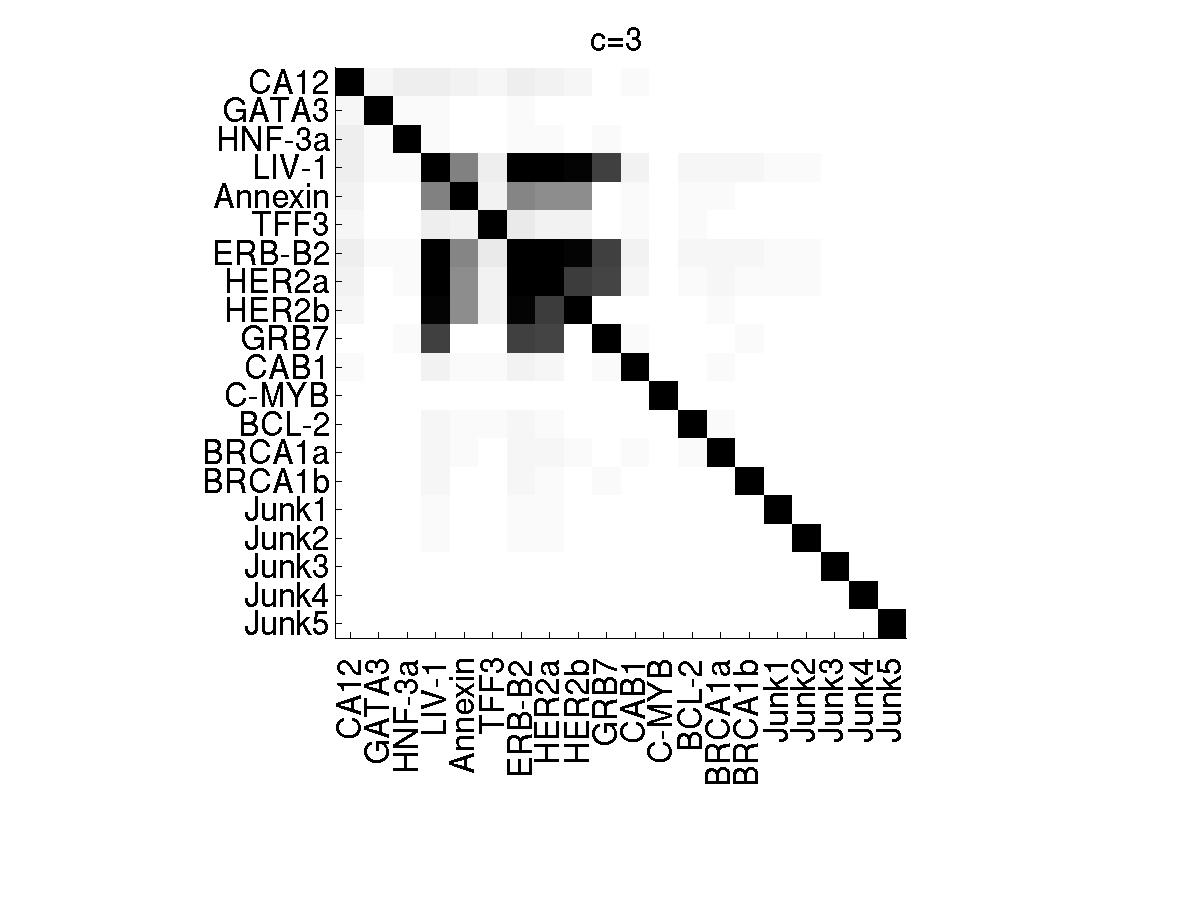} &
\hspace{-2.0cm}\includegraphics[width=0.5\textwidth]{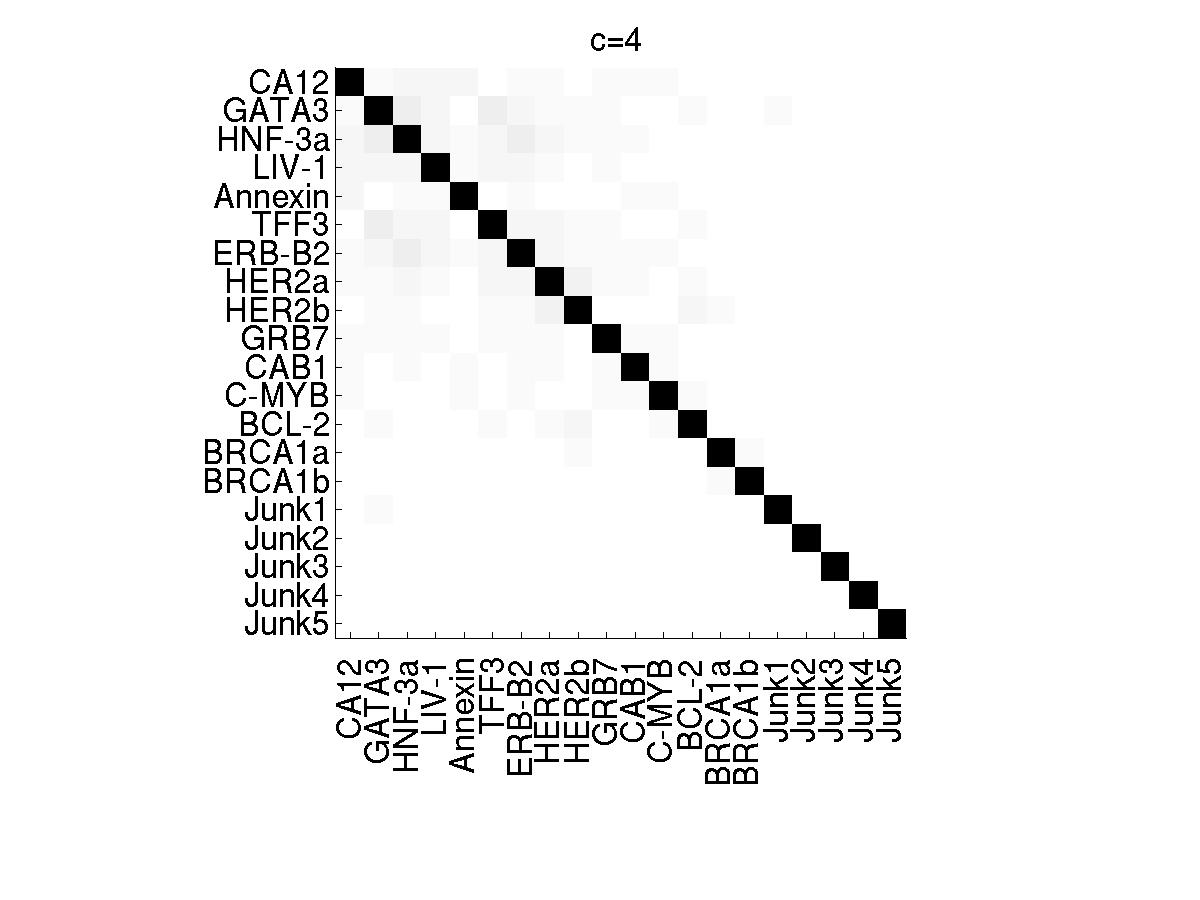}
\end{tabular}
\caption{Heat maps showing posterior probabilities of non-zero entries in $\K_c$ for each normal
mixture component $c=1:4$ in analysis of cancer gene expression data.  Shading runs from
white$=$0 to black$=$1 in each.
 \label{fig:postKedges}}
\end{figure}

\begin{figure}[bp!]
\centering
\phantom{.}\vskip0.75in
\begin{tabular}{rl}
\hspace{-0.2cm}\includegraphics[width=0.4\textwidth]{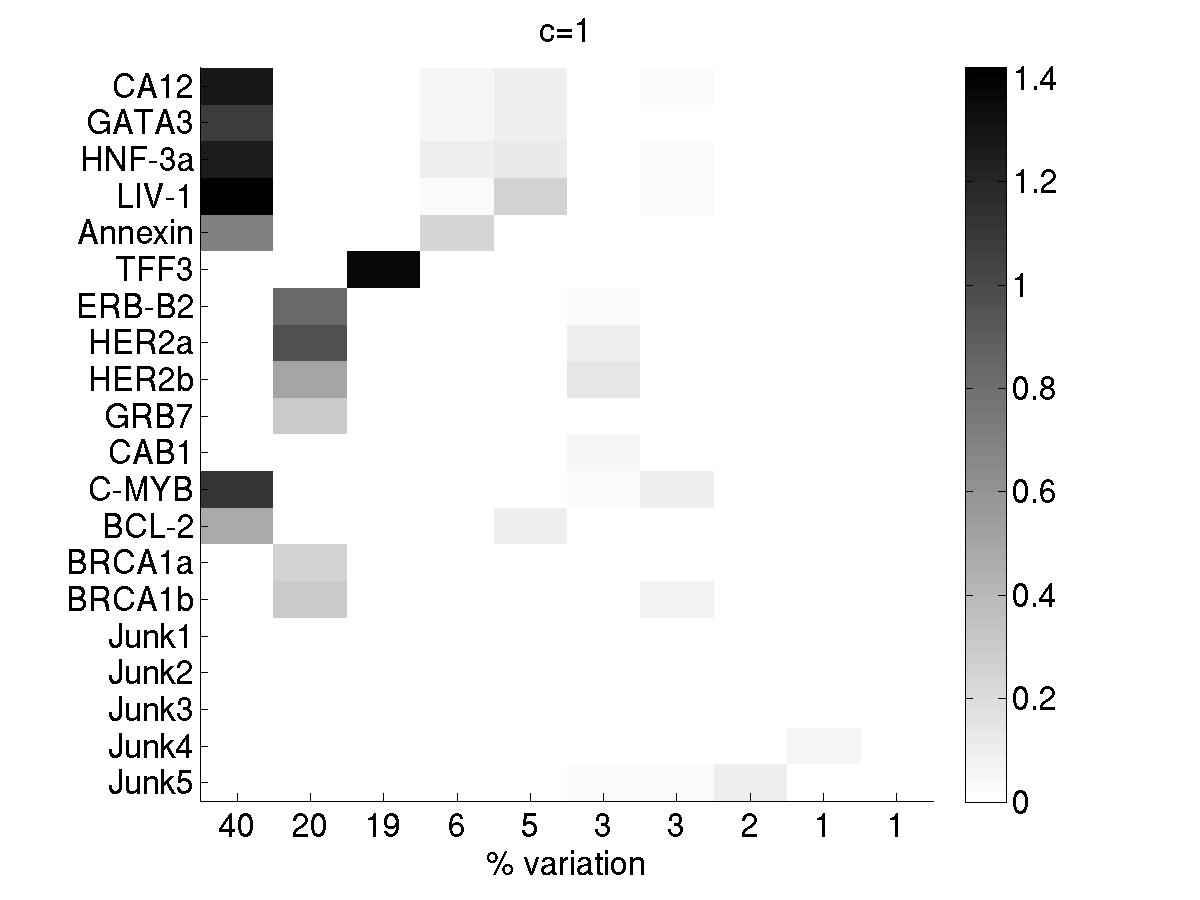} &
\hspace{-1.0cm}\includegraphics[width=0.4\textwidth]{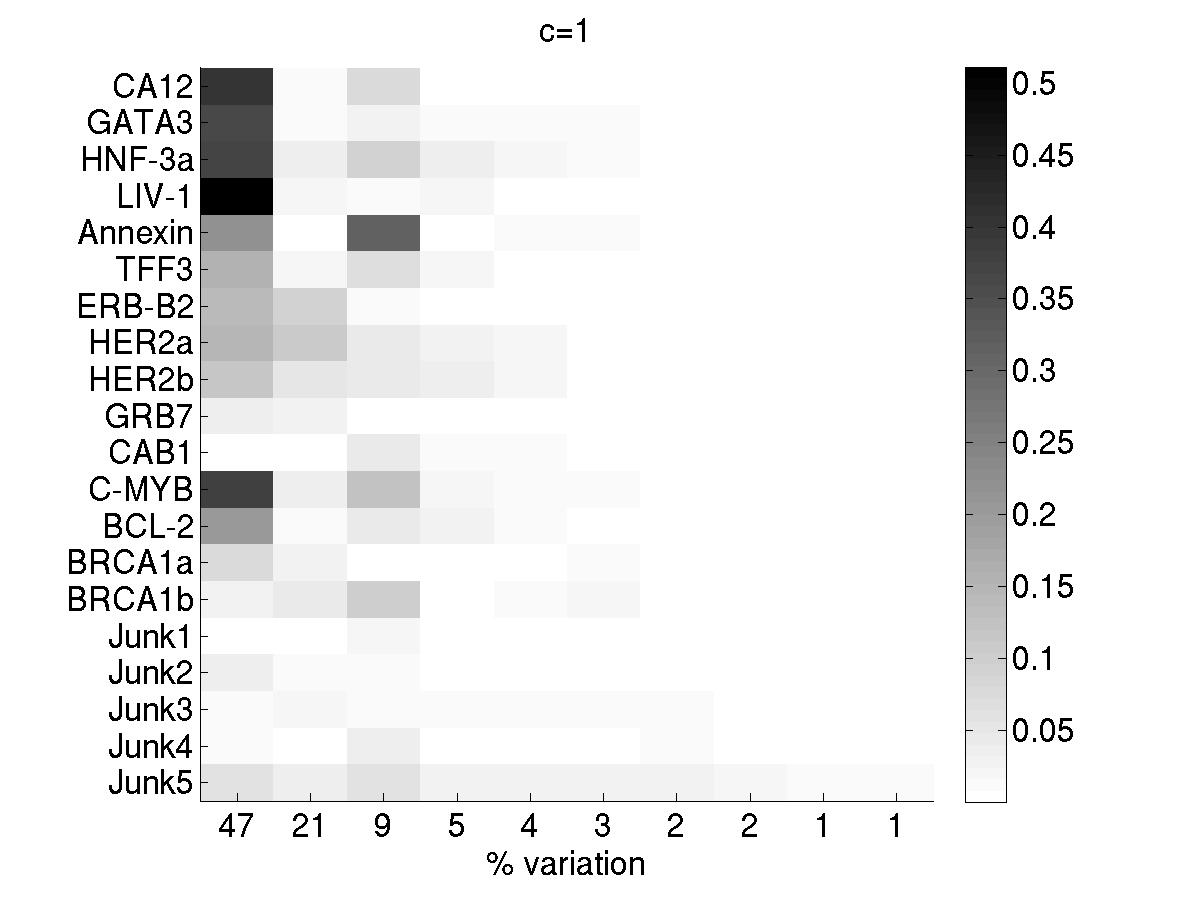} \\
\hspace{-0.2cm}\includegraphics[width=0.4\textwidth]{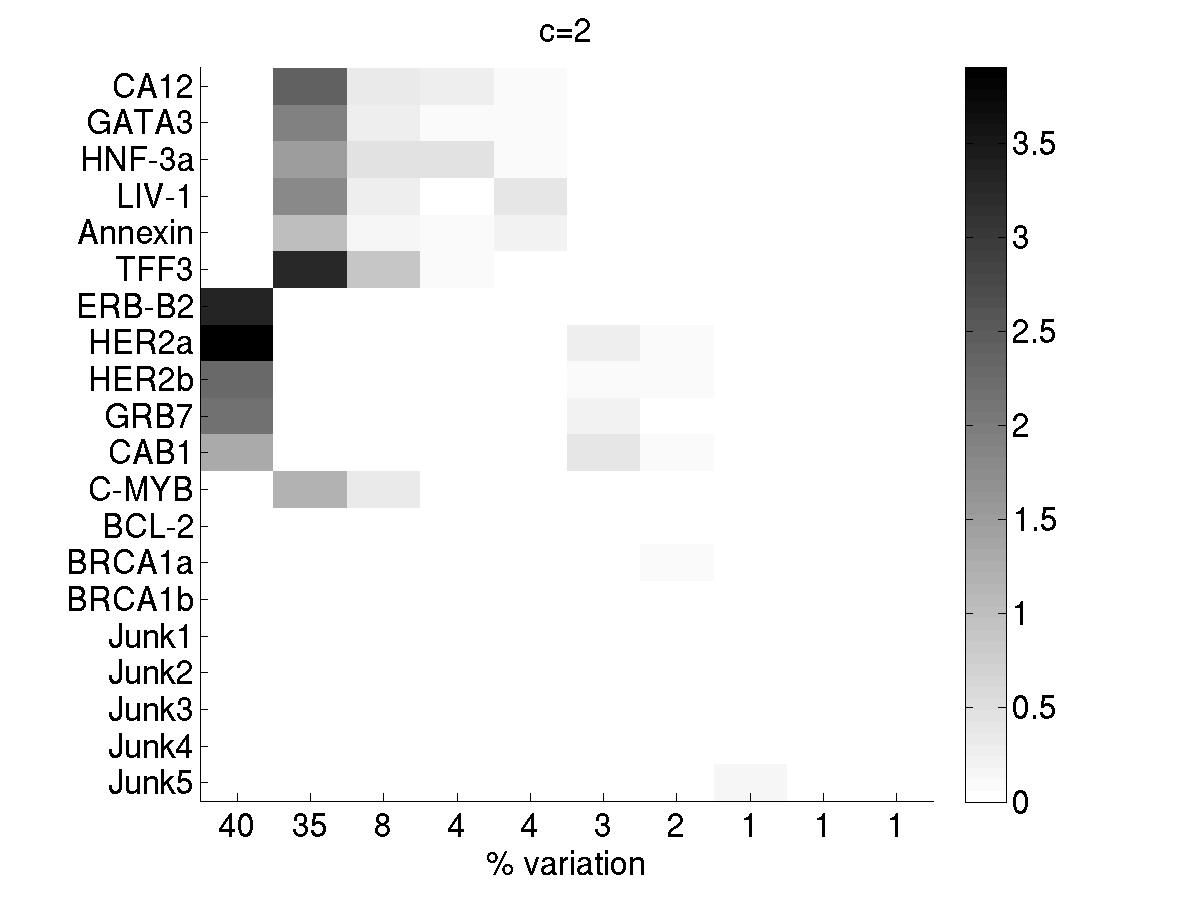} &
\hspace{-1.0cm}\includegraphics[width=0.4\textwidth]{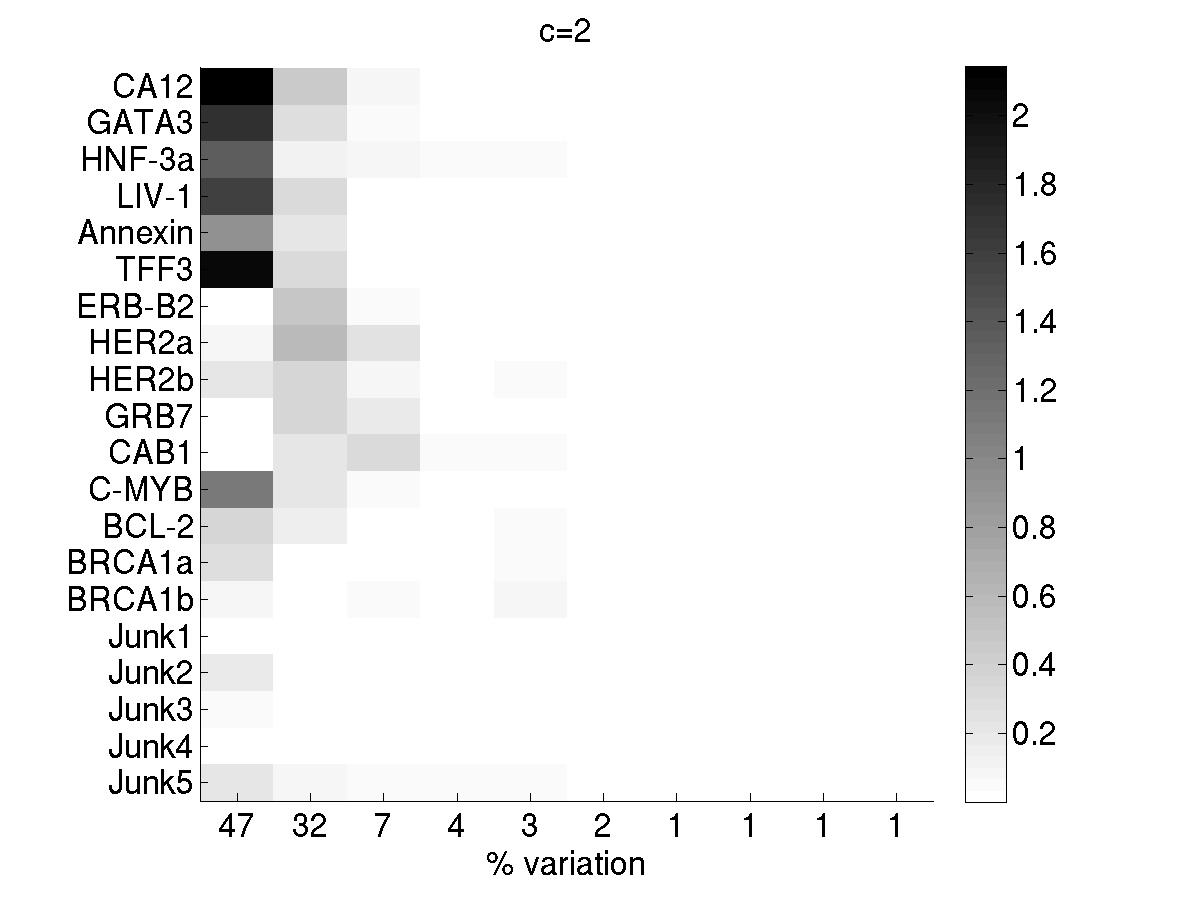}\\
\hspace{-0.2cm}\includegraphics[width=0.4\textwidth]{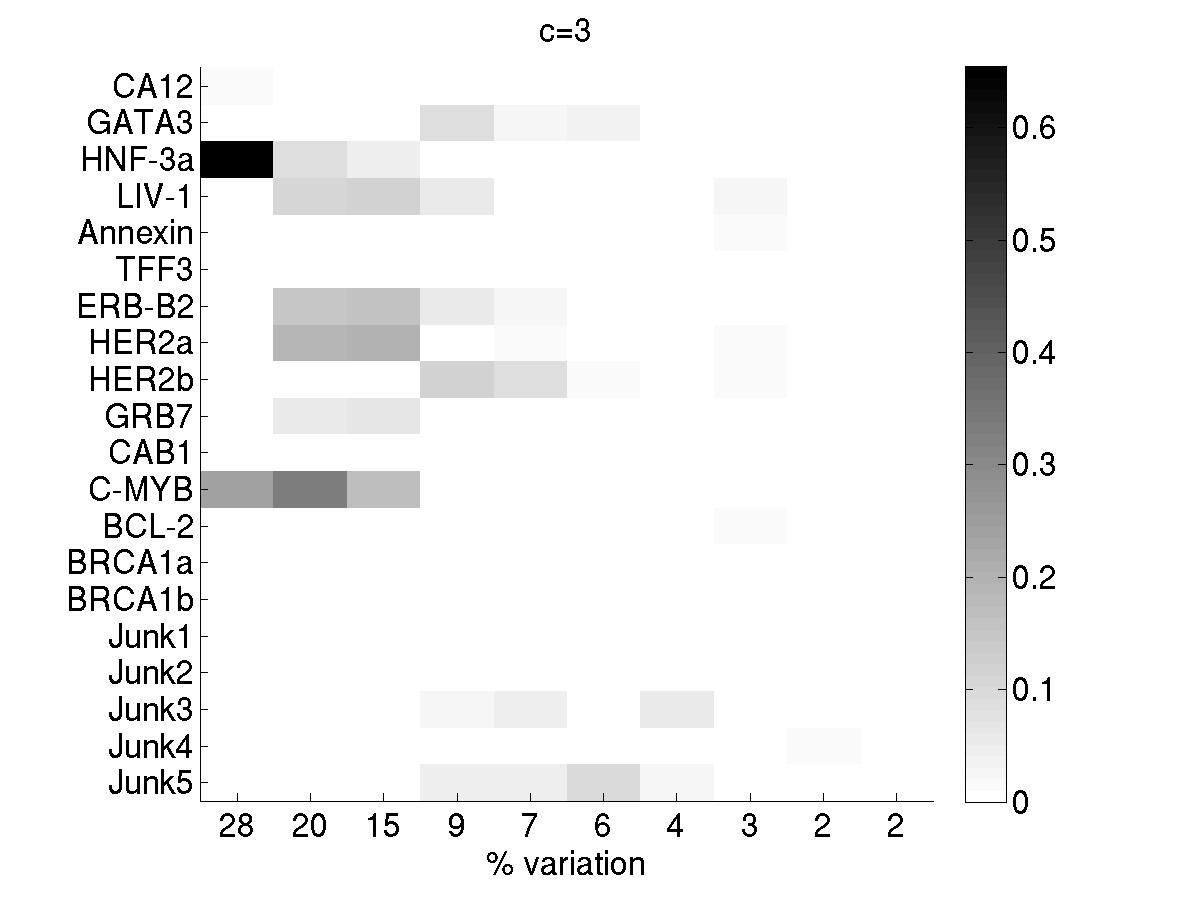} &
\hspace{-1.0cm}\includegraphics[width=0.4\textwidth]{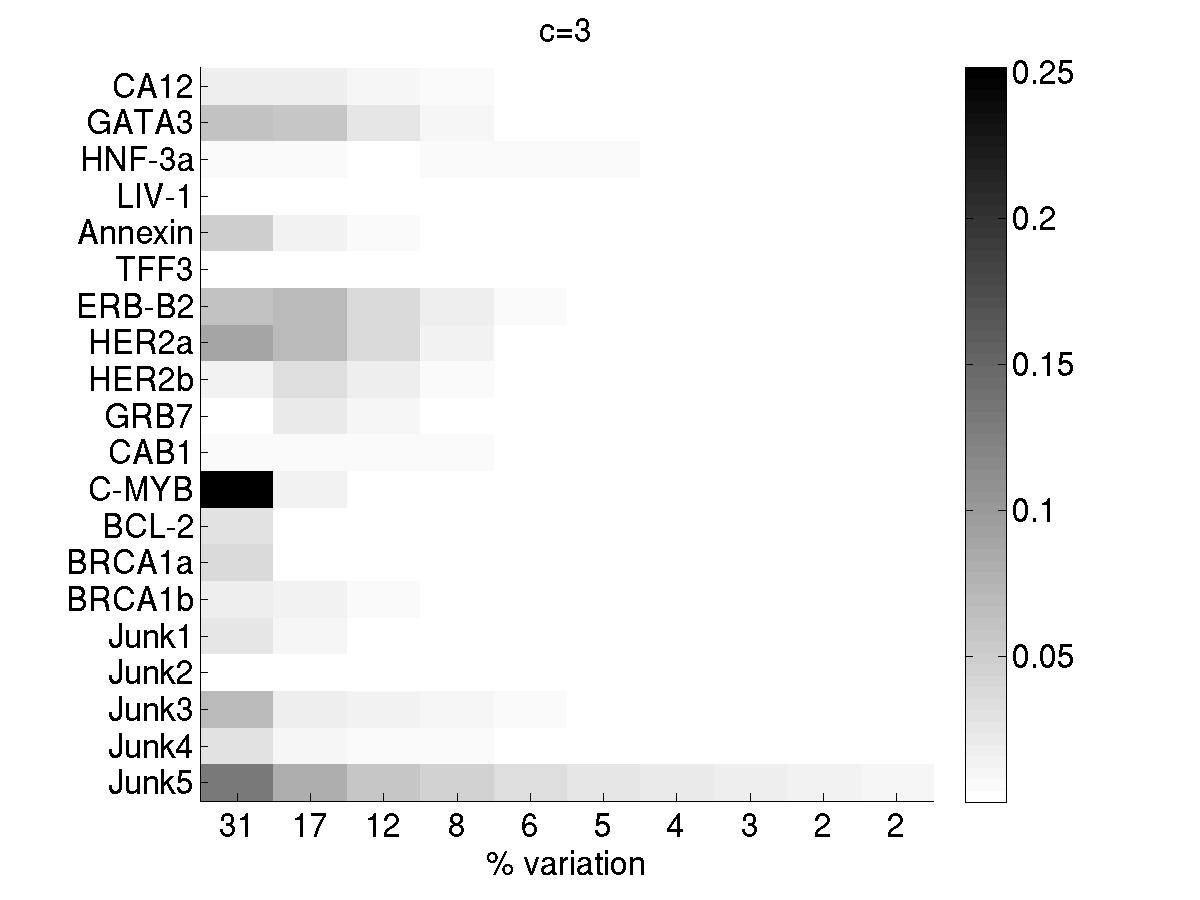}
\end{tabular}
\caption{Heat maps showing posterior means of the first 10 columns of $ \R_c\D_c^{1/2}$ for the three
main mixture components in analysis of cancer
 gene expression data. Left column:  Sparsity model,  and Right column: full non-sparse model.
The percent variation explained by each eigenvector is indicated
on the x-axes, based on posterior means of the $d_j$ in each case.
 \label{fig:postEmean}}
\end{figure}

Finally, we see that, appropriately, the low probability component $c=4$
has really no structure at all, consistent with prior for a basically empty component.
These conditional dependencies, and independencies, are better understood in terms of
the estimated factor structure underlying  eigenmatrices and eigenvalues of the
sparse Givens models in each component; these are shown in the left column of Figure~\ref{fig:postEmean}
for the three main components.
For the \lq\lq high ER" tumors in component $c=1,$  we see one dominant and two subsidiary eigenvectors, indicating
three \lq\lq ER-related factors"  based on non-zero loadings of the ER-related
genes; these presumably reflect several dimensions of the underlying patterns of variability in these
genes as a result of the complexity of the ER network.  The second dominant eigenvector relates to the
Her2 cluster.    For the \lq\lq high Her2" tumors in $c=2,$ we see the dominant factor is indeed linked to the
Her2 gene cluster, while the fact that ER related genes vary across the scale in these tumors
leads to a natural set of three or four ER-related factors.   For the triple negative/basal-like tumors in
component $c=3,$ we see residual biological pathway activity highlighted involving HNF-3$\alpha$ and C-MYB genes, as well as important factors in both ER and Her2 pathways; although these two pathways are less active in
tumors in this group, there is still meaningful variation among subsets of some of these genes.

For comparison, the right column in Figure~\ref{fig:postEmean} shows the corresponding eigenstructure
extracted from an analysis using traditional inverse Wishart priors on the $\bSigma_c$, i.e., in the standard
analysis with no sparsity.   It is very clear how the sharp factor-based groupings in the sparse Givens mixture
model \lq\lq cleans-up" the much noisier  standard results. In addition to cleaner and focused inference
on dependency structures, we also found that the standard analysis-- by comparison with the sparse model--
generates over-diffuse estimates of the spread of mixture components and so less sharp classification of
samples, as a result.

\section{Additional Comments
\label{sec:end}}

In terms of modelling variations and extensions, one interesting question relates to the
interpretation of the sparse Givens model as factor analysis. Our examples
have stressed this interpretation from an applied viewpoint. Theoretically, the Givens
 model is a full-rank, orthogonal factor analysis model. We can imagine extensions to include
reduced rank approximations that would be based on the use of priors giving
positive probability to zero values among the $d_i,$  relating more directly to
alternative factor modelling frameworks~\cite{Yoshida2010}.

We have experience in running the MCMC analysis for higher-dimensional
variance matrices, including extensions of the gene expression examples with $q=300$
genes. The overall performance of the MCMC is scalable, in terms of acceptance rates,
while of course the running time and implementation overheads increase.  In particular,
as the number of rotators grows, a number of computational challenges arise. First,
the numerical optimization to define Metropolis proposals $g_c(\omega_{i,j})$ becomes
increasingly time consuming, so that one immediate area of research will be to explore
more computationally efficient proposal strategies for the MCMC.
Second,  based on our  positive experience with the exploratory analysis to define
ad-hoc starting values for increasingly high-dimensional problems, one direction
 for improving the MCMC would be to consider alternatives to the
birth/death strategy based on more aggressive local search in neighborhoods of
\lq\lq good" sparsity configurations. Some of the concepts and
computational strategies underlying shotgun stochastic search in
regression and graphical models~\cite{Jones2005a,Hans2006,Hans2007a,Hans2007b}
may be of real benefit here. The potential for distributed
computation, including using GPU hardware~\cite{GPUisba2010,Suchard2010,GruberWest2015BA} is also of interest.

\subsubsection*{Supplementary Material}
As noted in the text, code implementing the analyses reported  here is (freely) available to
interested readers at the authors' web site.

\newpage

\section*{Appendix:  MCMC in Mixtures of Sparse Givens Models \label{sec:appendix}}
\addcontentsline{toc}{section}{Appendix}

Additional technical details of the MCMC algorithm in Section~\ref{sec:mixtures} are given here.

\medskip\noindent{(a) \em Starting values:}
We use $k-$means clustering to define initial, crude classification of the data into $C$ groups,
giving starting values for component indicators $\bgamma.$
Group means and proportions define starting values
$\bmu_c$ and $\w_c.$  Initial values for the Givens structures within each group are then
created using  the exploratory algorithm of
Section~\ref{sec:explore}. Beginning with the sample variance matrix of each group $c,$
this algorithm produces a sparse Givens structure with
starting values for the rotator pairs, angles and eigenvalues,
and hence $\R_c$, $\D_c$ and $\bSigma_c.$   The measurement error
variances $\psi_j$ in $\bPsi$ are initialized at draws from the prior.

\medskip\noindent{(b) \em Rotator structure and angle updates:}
For each cluster $c=1:C$ defined at the current iterate of the MCMC,
we update the rotators selected and corresponding angles using the RJ-MCMC analysis of
Section~\ref{sec:RJMCMC}.

\medskip\noindent{(c) \em Latent data $\X$:}
Each $\x_i$ is resampled from the complete conditional normal posterior whose
mean vector $\m_i$ and variance matrix $\M_i$ are given by
$$ \m_i = \M_i(\bPsi^{-1}\y_i + \bSigma_{\gamma_i}^{-1}\bmu_{\gamma_i})
\quad\textrm{and}\quad
\M_i^{-1}=\bPsi^{-1} + \bSigma_{\gamma_i}^{-1}.$$
Note that the $\bSigma_c^{-1}$ can
be calculated trivially even in high dimensional cases simply by
inverting the eigenvalues.

\medskip\noindent{(d) \em Measurement error variances $\psi_j$:}
Each of the $q$ elements of the diagonal matrix $\bPsi$ is resampled from a complete conditional
given by
\[ \psi_j^{-1} \sim Ga(\phi_a+n/2,\phi_b + \sum_{i=1}^n \epsilon_{ji}^2/2)  \]
for $j=1:q,$  where $\epsilon_{ji}$ is the $j-$the element of $\y_i-\x_i.$

\medskip\noindent{(e) \em Component indicators $\bgamma$:}
The set of $n$ component classification indicators $\bgamma=(\gamma_1,\ldots,\gamma_n)'$
are drawn from conditionally independent multinomials,
each with sample size 1 and probabilities over the $C$ cells defined by
\[ Pr(\gamma_i=c|-)\propto
w_cN(\x_i|\bmu_c, \bSigma_c), \qquad c=1:C,\]
where $N(\x|\cdot,\cdot)$ denotes the multivariate normal pdf.

\medskip\noindent{(f) \em Component weights $\w$:}
Resampled weights come from the complete conditional Dirichlet posterior with
parameter $(\alpha_1,\ldots,\alpha_C)'$ and
$\alpha_c =  1/C + n_c $ where $$ n_c = \sum_{i=1}^n I(\gamma_i=c) , \quad c=1:C. $$

\medskip\noindent{(g) \em Component means $\bmu_c$:}
Denote by $ \bar{\x}_c$ the sample mean in group $c$ given a current set of component indicators.  Then the
component means are sampled in parallel from the $C$ conditional normal posteriors with means
$(n_c+1/\tau)^{-1}n_c\bar{\x}_c$ and  variance matrices
$(n_c+1/\tau)^{-1}\bSigma_c,$ $c=1:C.$

\medskip\noindent{(h) \em Eigenvectors $\D_c:$}
Finally, the complete conditional distributions of the diagonal elements of $\D_c = \A_c^{-1}$
are independent inverse gammas constrained by the ordering; see \eqref{eq:postaj}, applied
to each of the $C$ groups in parallel. These are sampled in sequence using the inverse cdf method.

\end{document}